\documentclass[
  aps,
  pra,                    %
  reprint,                %
  amsmath,
  amssymb,
  longbibliography,
  citeautoscript,
  superscriptaddress,
  nofootinbib,
  floatfix,
  eqsecnum
]{revtex4-2}

\usepackage{graphicx}
\usepackage{epstopdf}
\usepackage{empheq}

\usepackage{mathtools}
\usepackage{amsthm}

\usepackage{stmaryrd}
\usepackage{mathrsfs}
\usepackage{latexsym}
\usepackage{verbatim}

\usepackage{enumitem}
\usepackage{microtype}
\usepackage{physics}
\usepackage{booktabs}

\usepackage[normalem]{ulem}
\usepackage{xcolor}

\usepackage{xurl}

\usepackage{tikz}
\usetikzlibrary{arrows.meta,calc,angles,quotes}

\definecolor{hardred}{RGB}{180,0,0}

\definecolor{ctp-green}{HTML}{40A02B}
\definecolor{ctp-mauve}{HTML}{8839EF}
\definecolor{ctp-blue}{HTML}{1E66F5}
\definecolor{ctp-teal}{HTML}{179299}
\definecolor{ctp-peach}{HTML}{FE640B}

\usepackage[
  colorlinks=true,
  breaklinks=true,
  citecolor=ctp-mauve,
  linkcolor=ctp-blue,
  urlcolor=ctp-blue,
  pdfborder={0 0 0},
  pdfborderstyle={/S/U/W 0},
]{hyperref}

\providecommand{\doi}[1]{}
\renewcommand{\doi}[1]{%
  \href{https://doi.org/#1}{\textcolor{ctp-blue}{doi:\nolinkurl{#1}}}%
}

\usepackage{orcidlink}

\DeclareMathAlphabet{\mathcal}{OMS}{cmsy}{m}{n}
\SetMathAlphabet{\mathcal}{bold}{OMS}{cmsy}{b}{n}

\DeclareMathOperator{\Arg}{Arg}
\DeclareMathOperator*{\argmax}{arg\,max}

\newcommand{\Var}{\mathrm{Var}}

\newtheorem{theorem}{Theorem}[section]
\newtheorem{lemma}[theorem]{Lemma}
\newtheorem{corollary}[theorem]{Corollary}
\newtheorem{proposition}[theorem]{Proposition}

\theoremstyle{definition}

\theoremstyle{remark}

\newtheorem{assumptions}{Assumptions}

\begin{document}

\title{Full-Period Optical Phase Estimation with Heisenberg Scaling Using
  Displaced Squeezed States and Gaussian Measurements}

\author{Marco A. Rodríguez-García\orcidlink{0000-0003-1504-0526} }
%\affiliation{Center for Quantum Information and Control, Department of Physics
%  and Astronomy, University of New Mexico, Albuquerque, New Mexico 87131, USA}
\affiliation{Institute for Quantum Science and Technology, University of
  Calgary, Calgary, AB T2N 1N4, Canada}
\affiliation{Center for Quantum Information and Control, Department of Physics
  and Astronomy, University of New Mexico, Albuquerque, New Mexico 87131, USA}

\author{Luis Medina-Dozal\orcidlink{0000-0002-4695-5190} }
\affiliation{
Instituto de Ciencias F\'isicas, Universidad Nacional Aut\'onoma de M\'exico,
Avenida Universidad s/n, Col. Chamilpa, Cuernavaca, Morelos 62210, Mexico
}

\author{Francisco E. Becerra\orcidlink{0000-0002-2928-310X} }
\affiliation{Center for Quantum Information and Control, Department of Physics
  and Astronomy, University of New Mexico, Albuquerque, New Mexico 87131, USA}

\date{\today}

\begin{abstract}
  We propose two-stage optimized strategies for full-period optical phase
  estimation with single-mode Gaussian states and Gaussian measurements under a
  fixed energy constraint. In the first Stage (Stage~I), displaced squeezed
  probes and heterodyne measurements provide coarse localization of the phase to
  a window in the circle. In the second Stage (Stage~II), squeezed-vacuum probes
  with adaptive homodyne measurements perform efficient phase estimation inside
  the selected window. We derive a generalized Cramér-Rao bound for this family
  of two-stage Gaussian strategies, which contains the contribution from local
  parameter estimation in Stage~II plus an overshoot penalty from coarse
  localization errors in Stage~I. For $E\le 25$ photons and squeezing limited to
  $12\,\mathrm{dB}$, protocols using displaced squeezed states in Stage~I reduce
  the optimized two-stage bound relative to protocols using coherent states in
  Stage~I, and remain within a factor of $3$--$30$ of the idealized local
  squeezed-vacuum quantum Cramér-Rao bound.
\end{abstract}

\maketitle

\section[intro]{Introduction}

Optical quantum metrology leverages quantum states of light to estimate physical
parameters with high precision \cite{barbieri2022optical, pirandola2018advances,
  Polino2020}, with applications including interferometry \cite{Polino2020},
atomic clocks \cite{Ludlow2015}, time transfer \cite{Giovannetti2001,
  Komar2014}, force sensing \cite{lawrie2019quantum, Aspelmeyer2014}, and
gravitational-wave detection \cite{Tse2019, Caves1981}. In many photonic sensing
platforms, an optical probe field interacts with the system under study and the
sensing information is mapped to the phase of the field, making phase estimation
a centerpiece for optical sensing.

Optical phase estimation with coherent states yields errors with variances
scaling as $1/\bar n$ in the mean photon number $\bar n$, defining the Standard
Quantum Limit (SQL) \cite{Paris2009, Giovannetti2011}. In contrast, optical
probes with quantum correlations, such as squeezed Gaussian states, can surpass
the SQL and achieve Heisenberg-type scaling $\sim 1/\bar n^2$ within an
identifiable phase interval. \cite{Caves1981, Monras2006, Giovannetti2011,
  RodrguezGarca2024}. Specifically, squeezed vacuum states exhibit an inherent
$\pi$-rotation symmetry \cite{RodrguezGarca2024}, in principle allowing phase
estimation with Heisenberg scaling within phase intervals of length $\pi$, when
optimal measurements over this range are available. However, in many sensing
applications, including carrier-phase navigation \cite{Teunissen1995IAR},
magnetometry \cite{arai2018geometric, degen2017quantum}, and optical atomic
clocks \cite{kaubruegger2021quantum, li2022improved}, no prior phase information
is available and the phase must be estimated globally over the full circle $\phi
\in [0, 2\pi)$ \cite{higgins2007entanglement, xiang2011entanglement,
  kolodynski2010phase}.

Recent advances in quantum measurements for phase estimation with squeezed
vacuum states have shown that adaptive homodyne strategies asymptotically attain
the quantum Cramér-Rao bound within an identifiable interval of length $\pi/2$,
with efficient performance in the finite-sample regime \cite{RodrguezGarca2024}.
This restriction follows from the $\pi$-inversion symmetry in the distribution
of homodyne outcomes, which, combined with the inherent $\pi$-periodicity of
squeezed vacuum, reduces the identifiability range to $\pi/2$ for locally
unbiased estimators. By incorporating a first stage of heterodyne measurements
for coarse phase localization, this adaptive strategy can be extended to the
full range $[0, \pi)$, which is the maximum interval over which squeezed vacuum
unambiguously encodes the phase \cite{Monras2006, Pinel2013, Olivares2009},
while maintaining asymptotically optimal performance \cite{RodrguezGarca2024}.
Recent experiments confirm that adaptive Gaussian schemes on squeezed light can
reach sub-SQL precision over this $[0,\pi)$ range \cite{Minati2025}.

Full-period phase estimation, however, requires quantum states beyond squeezed
vacuum that are identifiable over the full circle. A related approach optimizes
the Gaussian quantum probe (displacement, squeezing strength, and angle) at
every adaptive step, enabling full-period phase estimation beyond the SQL
\cite{RavellRodrguez2025}. However, the experimental and processing complexity
required for its implementation becomes daunting and impractical for current
experiments.

Here, we propose a practical two-Stage Gaussian protocol for full-period phase
estimation with two fixed, predetermined families of probe states and where
adaptivity is constrained to Gaussian measurements. Stage~I uses displaced
squeezed probes and heterodyne measurement to realize coarse localization on
$[0,2\pi)$ to an interval of length $\pi/2$. The nonzero displacement in Stage~I
makes restores full-period identifiablility, while squeezing increases the local
information extracted by heterodyne. Stage~II uses squeezed vacuum and adaptive
homodyne to perform locally efficient estimation inside the selected interval.

We derive a generalized Cram\'er-Rao bound for this class of two-Stage
strategies, in which the error bound decomposes into two contributions: the
local Stage~II error on the event that Stage~I interval contains the true phase,
and an overshoot penalty when Stage~I selects an incorrect interval. The
overshoot term has no counterpart in purely local Cram\'er-Rao theory. We
observe that optimizing the energy allocation between Stage~I and II under a
total energy constraint and a maximum allowable squeezing allows for suppressing
overshoot errors in Stage~I and improving local estimation in Stage~II. We find
that a modest amount of squeezing for coarse localization in Stage~I reduces the
final bound relative to coherent states. Our numerical studies show that these
two-Stage protocols implement full-period phase estimation within a factor of
$3$--$30$ of the idealized local QCRB with squeezed vacuum for moderate total
energies ($E\le 25$ photons) and squeezing up to $12$\,dB, and approach this
bound as the total available energy increases.

The article is organized as follows. Section~\ref{sec:preliminaries} provides
background on single-mode Gaussian states and measurements.
Section~\ref{sec:two_stage_strategy} introduces the two-Stage protocol and
derives the generalized Cramér-Rao bound. Sections~\ref{sec:coherent_init}
and~\ref{sec:dsvs_init} analyze the localization in Stage~I with coherent and
displaced squeezed probes, respectively.
Section~\ref{sec:two_stage_fixed_budget} describes the proposed optimized
two-Stage under fixed total energy. Section~\ref{sec:summary_outlook} presents a
discussion of possible extensions and describes our conclusions.

\section{Optical Phase Estimation with Gaussian States and Measurements}
\label{sec:preliminaries}

This section describes the preliminaries for the problem of full-period optical
phase estimation with single-mode Gaussian probes and Gaussian measurements. In
this global estimation problem the optical phase can take any value
$\theta\in[0,2\pi)$, with the endpoints understood modulo $2\pi$. Thus,
$[0,2\pi)$ is used as a set of representatives for the circle $\mathbb S^1$.
Here, we introduce the basic concepts for two-Stage protocols with conventional
Gaussian measurements, homodyne and heterodyne, described in
Sec.~\ref{sec:two_stage_strategy}. We discuss the phase-shift model, the finite
energy constraint, Gaussian quantum states and measurements, circular loss
function, and the corresponding Quantum and Classical Fisher information.

\subsubsection{Phase-shift model, resource constraint, and circular loss}

Consider a probe quantum state $\rho$ with finite energy on the single-mode
bosonic Fock space. We encode the unknown optical phase $\theta \in [0, 2 \pi)$
through the unitary representation $U_\theta=e^{-i\theta\hat n}$ of the circle
group $U(1)\simeq \mathbb S^1$, where $\hat n=a^\dagger a$ is the photon-number
operator. The phase-shifted family of quantum states under this transformation
becomes
\begin{equation}
  \rho_\theta \coloneqq U_\theta\,\rho\,U_\theta^\dagger,\qquad \theta\in[0,2\pi).
  \label{eq:phase_unitary_family_CH2}
\end{equation}
For infinite dimensional systems, such as Gaussian states, the corresponding
resource available for phase estimation corresponds to the energy, or
equivalently the mean photon number $\bar n(\rho) = \Tr(\rho\,\hat n)$. For
$N$ uses of independent probes, quantum model becomes
$\rho_\theta^{\otimes N}$ with total energy $N\,\bar n(\rho)$.

To quantify the statistical performance of an estimator $\hat\theta$, we measure
its risk through a loss function built from an appropriate notion of distance
between the estimates produced by $\hat\theta$ and the true parameter $\theta$.
Because $\theta$ is circular, the principal representative of an angle becomes
\begin{equation}
  \operatorname{wrap}(\varphi)
  \coloneqq
  \Arg\!\left(e^{i\varphi}\right)
  \in(-\pi,\pi],
  \label{eq:wrap_definition}
\end{equation}
with a corresponding signed wrapped estimation
error $
%\begin{equation}
  \varepsilon\left(\hat\theta,\theta \right)
  \coloneqq
  \operatorname{wrap}(\hat\theta-\theta)
  =
  \Arg\!\left(e^{i(\hat\theta-\theta)}\right)
%  \label{eq:signed_wrapped_error}
%\end{equation}
$, yielding a geodesic distance
\begin{equation}
  d\left(\hat\theta,\theta\right)
  \coloneqq
  \left|\operatorname{wrap}(\hat\theta-\theta)\right|
  =
  \min_{k\in\mathbb Z}
  \left|\hat\theta-\theta+2\pi k\right|.
  \label{eq:geodesic_distance}
\end{equation}
Thus, $d\left(\hat\theta,\theta \right)\in[0,\pi]$, whereas
$\varepsilon\left(\hat\theta,\theta \right)\in(-\pi,\pi]$. Finally, for an
estimator $\hat\theta$ constructed from the measurement outcomes of
$\rho_\theta$, we define the circular mean-square error by
\begin{equation}
  \mathrm{MSE}\left(\hat\theta;\theta\right) =
  \mathbb E_{\theta}\!\left[
    d\left(\hat\theta,\theta\right)^2
  \right].
  \label{eq:circular_mse}
\end{equation}
Here $\mathbb E_\theta$ denotes expectation with respect to the outcome
distribution induced by the chosen measurement under the parameter value
$\theta$.

\subsubsection{Measurements, induced statistical models, and Fisher information}

The amount of information about $\theta$ that can be extracted from the family
of parametrized quantum states $\{\rho_\theta\}$ is characterized by the
classical Fisher information, which depends on the classical statistical model
for the measurement outcomes. We introduce some basic concepts in statistics to
quantify this information.

Denote $\mathcal H$ the Hilbert space of the quantum system $\rho_\theta$. Let
$\mathcal B(\mathcal H)_+$ be the cone of positive bounded operators on
$\mathcal H$, and $(\mathcal X,\mathcal A)$ be the measurable space of possible
measurement outcomes. Here $\mathcal X$ is the set of outcomes, and $\mathcal A$
is the $\sigma$-algebra of events to which probabilities are assigned. A POVM is
a map $M:\mathcal A\to\mathcal B(\mathcal H)_+$ such that $M(\mathcal X)=\mathbb
I_{\mathcal H}$, where $M$ is countably additive in the weak operator topology.

For each value of the parameter $\theta\in\Theta\subseteq\mathbb R$, the
measurement induces a probability measure $\mathbb P_\theta^M$ on $(\mathcal
X,\mathcal A)$ through the Born rule
\begin{equation}
  \mathbb P_\theta^M(A)
  =
  \Tr\!\left(\rho_\theta M(A)\right),
  \qquad A\in\mathcal A.
  \label{eq:born_rule_povmCh1}
\end{equation}
Here $\mathbb P_\theta^M(A)$ is the probability that the measurement outcome
belongs to the event $A$ when the true parameter value is $\theta$. Assume that
the family $\{\mathbb P_\theta^M\}_{\theta\in\Theta}$ is dominated by a
$\sigma$-finite measure $\mu$ on $(\mathcal X,\mathcal A)$, and denote by
$p_\theta$ the Radon-Nikodym derivative of $\mathbb P_\theta^M$ with respect to
$\mu$, $p_\theta(\omega) = \frac{d\mathbb P_\theta^M}{d\mu}(\omega).$
% Equivalently,
% \begin{equation}
%   \mathbb P_\theta^M(A)
%   =
%   \int_A p_\theta(\omega)\,\mu(d\omega),
%   \qquad A\in\mathcal A.
% \end{equation}
If $p_\theta(\omega)$ is differentiable at $\theta$ and satisfies the standard
regularity conditions \cite[Sec.~2.5, pp.~115--116]{Lehman1998}, the Fisher
information of the classical statistical model induced by the measurement $M$ is
\begin{equation}
  \mathcal I_M(\theta)
  =
  \mathbb E_\theta^M\!\left[
    \left(\partial_\theta \log p_\theta(\omega)\right)^2
  \right].
  \label{eq:classical_fi}
\end{equation}
The Fisher information $\mathcal I_M(\theta)$ quantifies the information about
$\theta$ contained in the outcome distribution generated by applying the POVM
$M$ to $\rho_\theta$. For $N$ independent repetitions of the same measurement on
identical copies of $\rho_\theta$, the Fisher information of the complete
quantum state $\rho_\theta^{\otimes N}$ is $N\,\mathcal I_M(\theta)$.

The classical Cramér-Rao bound (CRB) provides a lower bound on the variance of
any estimator that is locally unbiased at the true value of the parameter
$\theta$. This means that the estimator is unbiased at the true value and has
vanishing bias to first order in a neighborhood of that value. Given that the
phase is circular and this bound is a local and asymptotic statement, we
consider a local phase coordinate near a fixed value $\theta_0\in[0,2\pi)$. Let
$\tilde\theta_{\theta_0}$ denote the estimator expressed in such a local
coordinate, so that $\tilde\theta_{\theta_0}-\theta_0$ is the signed local
estimation error. We consider that the model obeys the usual regularity
conditions \cite[Sec.~2.5, pp.~115--116]{Lehman1998}, and suppose that
$\tilde\theta_{\theta_0}$ is locally unbiased at $\theta_0$ so that
\begin{equation}
  \mathbb E_{\theta_0}^M\!\left[
    \tilde\theta_{\theta_0}
  \right]
  =
  \theta_0,
  \qquad
  \left.
  \frac{d}{d\theta}
  \mathbb E_{\theta}^M\!\left[
    \tilde\theta_{\theta_0}
  \right]
  \right|_{\theta=\theta_0}
  =
  1.
\end{equation}
Under these assumptions the CRB provides the lower bound for the variance of
this estimator
\begin{equation}
  \mathrm{Var}_{\theta_0}^M\left(\tilde\theta_{\theta_0}\right)
  \ge
  \frac{1}{\mathcal I_M(\theta_0)}.
  \label{eq:crb_classical_local}
\end{equation}
An estimator that attains this bound at $\theta_0$ is said to be efficient at
$\theta_0$.

\subsubsection{Quantum Fisher information and the quantum Cramér-Rao bound}

More generally, we can define the quantum Cramér-Rao bound (QCRB), which is the
CRB optimized over all the possible measurements allowed by quanutm mechanics.
The QCRB is defined as the inverse of the quantum Fisher information (QFI)
$\mathcal F_Q^{\rho}(\theta) =\Tr(\rho_\theta L_\theta^2)$, where $L_\theta$ is
the symmetric Logarithmic derivative, defined by $\partial_\theta\rho_\theta =
\frac12\left(L_\theta\rho_\theta+\rho_\theta L_\theta\right),$ whenever such an
operator exists on the support of $\rho_\theta$. Moreover, for any POVM $M$, we
have $\mathcal{I}_M(\theta)\ \le\ \mathcal F_Q^{\rho}(\theta),$
\cite{Braunstein1994}. Therefore, any locally unbiased estimator $\hat \theta$
expressed in a local coordinate around $\theta_0$, satisfies the QCRB
\begin{equation}
  \mathrm{Var}_{\theta_0}^{M}\left(\hat\theta \right)\ \ge\ \frac{1}{\mathcal F_Q^{\rho}(\theta_0)}.
  \label{eq:qcrb_local}
\end{equation}

For a product of $N$ independent probes $\rho_\theta^{\otimes N}$, the QFI is
additive so that $\mathcal{F}_Q^{N}(\theta)=N\,\mathcal{F}_Q^{\rho}(\theta)=$.
Furthermore, when $\rho_\theta$ in Eq.~\eqref{eq:phase_unitary_family_CH2} is
pure, $\rho_\theta=\ket{\psi_\theta}\!\bra{\psi_\theta}$ with
$\ket{\psi_\theta}=U_\theta\ket{\psi}$, the QFI is independent of $\theta$
\cite{Braunstein1994}:
\begin{equation}
  \mathcal{F}_Q^{\rho}=4\,\mathrm{Var}_{\psi}\left(\hat n \right).
  \label{eq:qfi_pure_generator}
\end{equation}

\subsubsection{Gaussian states and measurements}
\label{subsub:homodyne}

In this work, we consider phase estimation strategies with Gaussian states and
measurements. Specifically, we consider probes in single-mode Gaussian states,
their transformation under phase shifts, and heterodyne and homodyne
measurements that define the statistical models described below.

Single-mode Gaussian states $\rho$ are characterized by a displacement vector $d
=\langle R\rangle_\rho\in\mathbb{R}^2$ and a covariance matrix
$V\in\mathbb{R}^{2\times 2}$:
\begin{equation}
	V_{jk} = \tfrac12\langle \{R_j-d_j,R_k-d_k\}\rangle_\rho,\qquad j,k\in\{1,2\},
	\label{eq:cov_def}
\end{equation}
subject to the uncertainty constraint $V+\tfrac{i}{2}\Omega\ge 0$ with
$\Omega=\left(\begin{smallmatrix}0&1\\-1&0\end{smallmatrix}\right)$
\cite{Weedbrook2012}. Here the vector $R =(q,p)^\top$, with $q =
\frac{a+a^\dagger}{\sqrt{2}}$ and $ p= \frac{a-a^\dagger}{i\sqrt{2}}$ the field
quadratures. A phase shift transformation $U_\theta$ of a Gaussian state
generates a rotation of the quadrature vector $R$ in the Heisenberg picture
\begin{equation}
  U_\theta^\dagger R\,U_\theta = \mathsf R(-\theta)\,R, \quad \mathrm{with} \quad   \mathsf R(\theta)=
  \begin{pmatrix}
  	\cos\theta & -\sin\theta\\
  	\sin\theta & \cos\theta
  \end{pmatrix}.
  \label{eq:heisenberg_rotation_quadratures}
\end{equation}
Then, an initial Gaussian state $\rho$ with first and second moments $(d,V)$
transforms into another Gaussian state $\rho_\theta$ under $U_\theta$ with
rotated moments $d_\theta = \mathsf R(-\theta)\,d$ and $V_\theta = \mathsf
R(-\theta)\,V\,\mathsf R(-\theta)^\top$, respectively.

Gaussian measurements, specifically homodyne and heterodyne, applied to Gaussian
states induce Gaussian probability distributions for the measurement outcomes.
For the parametrized state $\rho_\theta$, the corresponding measurement outcome
distribution is Gaussian with mean $m_\theta$ and covariance matrix
$\Sigma_\theta$, both depending differentiably on $\theta$. This property allows
us to calculate the Fisher information for $\theta$ associated with Gaussian
measurements:
\begin{equation}
  \mathcal{I}(\theta)
  = (\partial_\theta m_\theta)^\top \Sigma_\theta^{-1}(\partial_\theta m_\theta)
    + \tfrac12\Tr\!\left[ \left(  \Sigma_\theta^{-1}(\partial_\theta \Sigma_\theta) \right)^2
\right].
  \label{eq:fi_multivariate_normal_CH2}
\end{equation}
With this general model, we can obtain the Fisher information of $\theta$ for
any Gaussian measurement, including homodyne and heterodyne.

Specifically, the homodyne measurement, with a fixed local-oscillator phase
$\phi\in[0,2\pi)$, is a projection-valued measure (PVM) corresponding to the
quadrature operator onto the quadrature observable
\begin{equation}
  X_\phi \coloneqq q\cos\phi + p\sin\phi
  = \frac{1}{\sqrt{2}}\left(ae^{-i\phi}+a^\dagger e^{i\phi}\right).
  \label{eq:rotated_quadrature}
\end{equation}
We denote this PVM by $\Pi_\phi$, so that
\begin{equation}
  \Pi_\phi(B)=\mathbf 1_B(X_\phi),
  \qquad B\in\mathcal B(\mathbb R),
\end{equation}
where $\mathcal B(\mathbb R)$ denotes the Borel $\sigma$-algebra on $\mathbb R$.
For a Gaussian state $\rho_\theta$, this probability measure is Gaussian, and
the homodyne measurement outcomes satisfy
\begin{equation}
  X_\phi
  \sim
  \mathcal N\!\left(
    \mu_\theta^\phi,
    (\sigma_\theta^\phi)^2
  \right),
\end{equation}
with mean $\mu_\theta^\phi=e_\phi^\top d_\theta,$ and variance
$(\sigma_\theta^\phi)^2=e_\phi^\top V_\theta e_\phi,$ where
$e_\phi=(\cos\phi,\sin\phi)^\top.$ Using Eq.
(\ref{eq:fi_multivariate_normal_CH2}), we obtain the Fisher information about
$\theta$ for the homodyne measurement at phase $\phi$:
\begin{equation}
  \mathcal{I}_{\mathrm{hom}}(\theta;\phi)
  = \frac{\left(\partial_\theta \mu_\theta^\phi\right)^2}{\left(\sigma_\theta^\phi\right)^2}
  + \frac{\left(\partial_\theta \left(\sigma_\theta^\phi \right)^2\right)^2}{2\,\left(\sigma_\theta^\phi\right)^4}.
  \label{eq:fi_homodyne}
\end{equation}

In an similar way, we consider the heterodyne measurement, which is described
by the coherent-state POVM
\begin{equation}
  \Lambda(d^2\alpha)=\pi^{-1}\ket{\alpha}\!\bra{\alpha}\,d^2\alpha, \quad \alpha \in \mathbb{C}.
  \label{eq:coherent_POVM}
\end{equation}
The corresponding outcome probability density is given by the Husimi
$Q$-function \cite{Husimi1940-vy} $ Q_{\rho_\theta}(\alpha) \coloneqq
\frac{1}{\pi}\,\langle\alpha|\rho_\theta|\alpha\rangle$. Operationally, the
heterodyne measurement can be interpreted as the canonical joint (unsharp)
measurement of the conjugate quadratures $q$ and $p$
\cite{Leonhardt1995,Weedbrook2012}, yielding simultaneous noisy outcomes of both
quadratures. When convenient, we identify the complex heterodyne outcome
$\alpha\in\mathbb C$ with the real vector
$Y=(\sqrt2\,\mathrm{Re}\,\alpha,\sqrt2\,\mathrm{Im}\,\alpha)^\top$. Using this
real-quadrature notation
%, an input Gaussian state $\rho_{\theta}$ with moments $(d_\theta,V_\theta)$
%produces a Gaussian outcome with mean $d_\theta$ and covariance
%$V_\theta+\frac12\mathbb I_2$.
, and considering a heterodyne measurement of the single mode Gaussian state
$\rho_\theta$ defined by $(d_\theta, V_\theta)$, the distribution of measurement
outcomes is a 2-D Gaussian with mean $m_\theta=d_\theta$ and covariance
$\Sigma_\theta=V_\theta+\tfrac12\mathbb I_2$. The Fisher information about
$\theta$ for the heterodyne measurement is obtained from
Eq.~(\ref{eq:fi_multivariate_normal_CH2}) as:
\begin{equation}
\begin{aligned}
  \mathcal{I}_{\mathrm{het}}(\theta)
  &=
  (\partial_\theta d_\theta)^\top
  \!\left(V_\theta+\tfrac12\mathbb{I}_2\right)^{-1}
  (\partial_\theta d_\theta)  \\
  &\quad
  + \frac12
  \Tr\!\left[
    \left\{
      \left(V_\theta+\tfrac12\mathbb{I}_2\right)^{-1}
      \partial_\theta V_\theta
    \right\}^2
  \right].
\end{aligned}
\label{eq:fi_heterodyne}
\end{equation}

\subsubsection{Pure Gaussian probes}
\label{sec:gaussian_probes_scaling_symmetry}

Up to a global phase, any pure single-mode Gaussian quantum state can be written
as a displaced squeezed vacuum state, $\ket{\alpha,\zeta} =
D(\alpha)\,S(\zeta)\ket{0},$ with displacement amplitude $\alpha\in\mathbb{C}$
and squeezing parameter $\zeta=re^{i\psi}$, where $r\ge 0$ and
$\psi\in[0,2\pi)$. Here, $D(\alpha)=\exp(\alpha a^\dagger-\alpha^*a)$ is the
displacement operator, while
$S(\zeta)=\exp[\tfrac12(\zeta^*a^2-\zeta(a^\dagger)^2)]$ is the squeezing
operator \cite{Weedbrook2012}. Under the phase-shift transformation
$U_\theta=e^{-i\theta\hat n}$ these probe states become
\begin{equation}
  \ket{\alpha,\zeta;\theta} =U_\theta\ket{\alpha,\zeta},
  \qquad \theta\in[0,2\pi).
  \label{eq:phase_encoded_family_rewrite}
\end{equation}
Since $U_\theta$ commutes with $\hat n$, the expectation of the number operator
$ \bar n(\alpha,\zeta) \coloneqq \bra{\alpha,\zeta}\hat n\ket{\alpha,\zeta} =
|\alpha|^2+\sinh^2 r$ is constant along the encoded orbit
$\{\ket{\alpha,\zeta;\theta} \colon \theta\in[0,2\pi)\}$, and is additive across
independent uses of the probe state. Thus, we identify the mean photon number of
the input probe as the the metrological resource.

\subsubsection{Local sensitivity performance and squeezed-vacuum symmetry}
\label{subsec:sensitivity}

We use the QFI $F_Q^{\rho}$ Eq.~\eqref{eq:qfi_pure_generator} to define the
maximum local sensitivity for a particular quantum probe state. More precisely,
we identify the SQL with a linear scaling of $\mathcal{F}_Q^{\rho}=O(\bar n)$,
whereas Heisenberg-scaling corresponds to a quadratic scaling
$\mathcal{F}_Q^{\rho}=\Theta(\bar n^2)$ \cite{Giovannetti2011,Paris2009}. In
particular, the SQL corresponds to the precision achieved by states with no
quantum correlations such as coherent states $\ket{\alpha}=D(\alpha)\ket{0}$.
Specifically, the QFI for coherent states with mean photon number $\bar
n=|\alpha|^2$ and $\Var_\alpha \left( \hat{ n} \right)=|\alpha|^2$ becomes
\begin{equation}
  F_Q^{\mathrm{coh}} = 4|\alpha|^2 = 4\bar n.
  \label{eq:qfi_coh_rewrite}
\end{equation}

At the opposite extreme, the QFI for squeezed-vacuum
states (SVS) $\ket{\zeta} =S(\zeta)\ket{0}$, for which $\bar n=\sinh^2 r$
and $\Var \left( \hat n \right)=2\bar n(\bar n+1)$ becomes
\begin{equation}
  F_Q^{\mathrm{SVS}}
  = 8\,\bar n(\bar n+1).
  \label{eq:qfi_svs_rewrite}
\end{equation}
showing with quadratic scaling with $\bar n$ corresponding to the
Heisenberg-scaling \cite{Monras2006,Pinel2013}.
% As a result, for a single squeezed-vacuum probe, or for a fixed number of
% probes with growing energy per probe, the QFI shows quadratic scaling with
% $\bar n$ \cite{Monras2006,Pinel2013}.
Noteworthy, adaptive homodyne strategies can exploit this local optimality and
approach the QCRB in the asymptotic limit of many repetitions, provided the
measurements are kept near the locally optimal operating point
\cite{RodrguezGarca2024}.

Nevertheless, a fundamental limitation of squeezed vacuum probes is their
$\pi$-rotation symmetry. Specifically, squeezed vacuum states are invariant
under a phase shift of $\pi$, $
%\begin{equation}
	\rho_{\theta+\pi}=\rho_\theta,
	 \text{for all }\theta.
	\label{eq:svs_pi_symmetry_rewrite}
%\end{equation}
  $ Thus, no measurement performed on squeezed-vacuum probes alone can
  distinguish $\theta$ from $\theta+\pi$. Moreover, beyond the state
  $\pi$-ambiguity, homodyne measurements introduce a further restriction at the
  level of the likelihood. For a squeezed-vacuum probe and a fixed
  local-oscillator phase, the homodyne outcome is Gaussian with zero mean and a
  variance depending on the phase through a second harmonic. The resulting
  likelihood is therefore $\pi$-periodic and is not globally injective. On a
  suitable interval of length $\pi/2$, the relevant variance dependence becomes
  one-to-one, which gives the local identifiability range used in the two-stage
  construction \cite{RodrguezGarca2024}. As a result, a full-period strategy
  cannot rely only on squeezed-vacuum probes and homodyne detection. It must
  either introduce a symmetry-breaking resource, such as a nonzero displacement,
  or assume external prior information that localizes the phase to an
  identifiable interval.

\section{Two-stage full-period optical phase estimation}
\label{sec:two_stage_strategy}

\begin{figure*}[t]
\centering
\includegraphics[width=18cm]{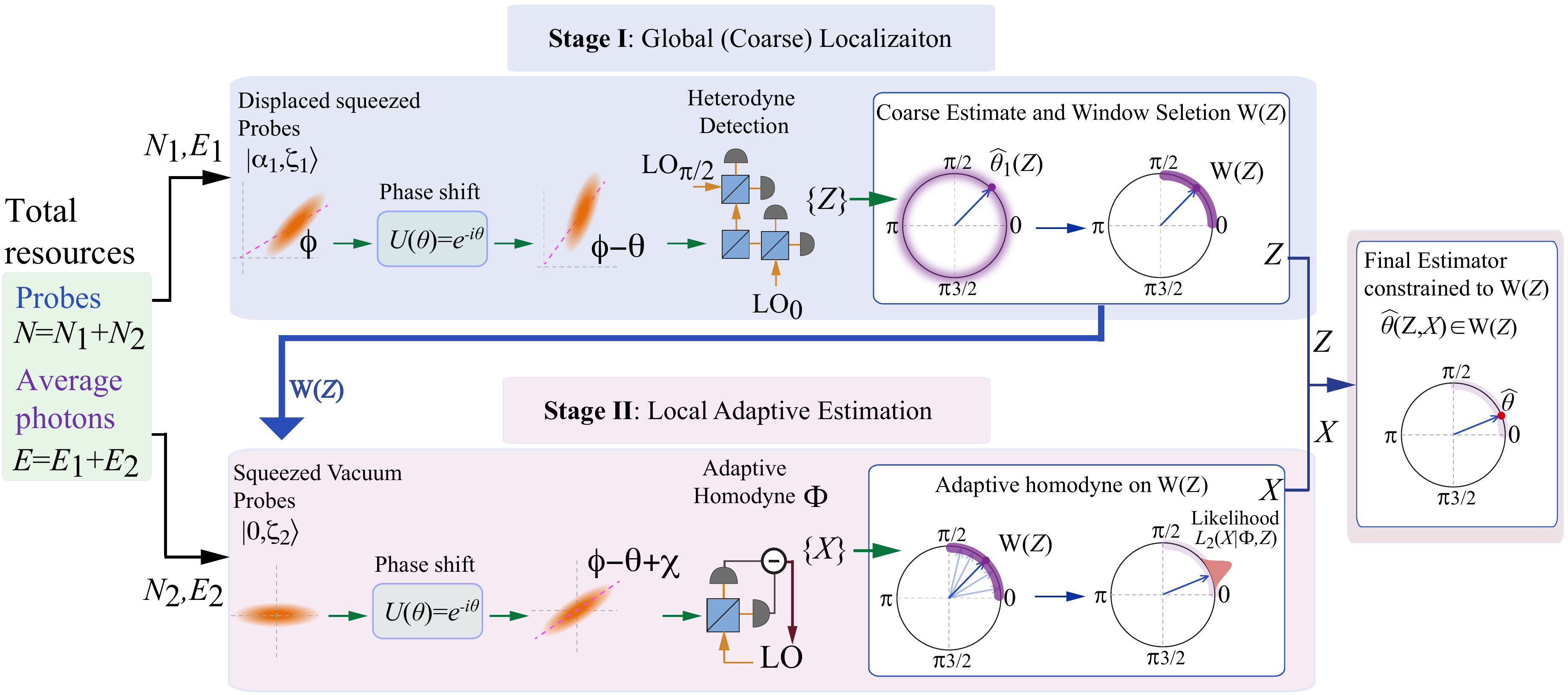} %\columnwidth
\caption{\label{fig:two_stage_protocol} Schematic representation of the proposed
  two-stage protocol for full-period optical phase estimation. A total budget of
  $N=N_1+N_2$ probes and $E=E_1+E_2$, measured in mean photon number units is
  divided between two stages. In Stage~I, $N_1$ displaced squeezed probes
  $\ket{\alpha_1,\zeta_1}$ are measured by heterodyne detection in order to
  produce a coarse phase estimate $\hat\theta_1$ and a selected identifiable
  interval $W(Z)$ of length $\pi/2$. The nonzero displacement removes the
  $\pi$-symmetry of squeezed vacuum and restores identifiability over the full
  circle. In Stage~II, $N_2$ squeezed-vacuum probes $\ket{0,\zeta_2}$ are
  measured by an adaptive homodyne strategy conditioned on $W(Z)$, yielding
  locally efficient estimation within the selected interval when that interval
  contains the true phase \cite{RodrguezGarca2024}. The final estimator
  $\hat\theta=\hat\theta(Z,X)$ is therefore constrained to lie in $W(Z)$. This
  decomposition separates the tasks of global localization and local
  high-precision estimation, and it is the basis for the generalized lower bound
  for two-Stage strategies derived in this work.}
\end{figure*}

Figure~\ref{fig:two_stage_protocol} shows the schematic of the proposed
Two-Stage strategy for full-period phase estimation based on Gaussian states and
measurement under a a total energy constraint $E=E_1+E_2>0$. The first Stage
(Stage~I) uses $N_1$ displaced squeezed probes $\ket{\alpha_1,\zeta_1}$ with
$|\alpha_1|>0$, and heterodyne detection to localize the unknown phase within an
identifiable interval of length $\pi/2$. Displaced squeezed states are
identifiable on the full circle, and heterodyne outcomes remain Gaussian with a
mean that rotates with the full phase period, maintaining this full-period phase
identifiability. Then, the second Stage (Stage~II) uses $N_2$ squeezed-vacuum
probes $\ket{0,\zeta_2}$ and implements an adaptive homodyne strategy that
enables asymptotically optimal phase estimation at the QCRB within the window of
length $\pi/2$ \cite{RodrguezGarca2024}.

Writing explicitly the squeezing parameter $\zeta_j=r_je^{i\psi_j}$ and given that $\bar
n(\alpha,\zeta)=|\alpha|^2+\sinh^2 r$, the total energy across all the probes is
\begin{equation}
  N_1\,\bar n(\alpha_1,\zeta_1)
  +
  N_2\,\bar n(0,\zeta_2)
  \le
  E.
  \label{eq:energy_budget_two_stage}
\end{equation}

After the phase-shift operation $U_\theta=e^{-i\theta\hat n}$, with
$\theta\in[0,2\pi)$, the induced family for the $N_1+N_2$ probes is the product
model $\rho_\theta^{(1)\otimes N_1}\otimes\rho_\theta^{(2)\otimes N_2}$. Here,
$\rho_\theta^{(j)}=U_\theta\rho^{(j)}U_\theta^\dagger$ for $j=1,2$, with
$\rho^{(1)}=\ket{\alpha_1,\zeta_1}\bra{\alpha_1,\zeta_1}$ and
$\rho^{(2)}=\ket{0,\zeta_2}\bra{0,\zeta_2}$.

Let $Z=(Z_1,\dots,Z_{N_1})$ denote the heterodyne outcomes in Stage~I with
$Z_k\in\mathbb R^2$. Based on $\{Z_k\}$, Stage~I returns two objects: a coarse
estimator $\hat\theta_1=\hat\theta_1(Z)\in[0,2\pi)$, and a localization window
of fixed length $\pi/2$ centered at $\hat\theta_1$,
\begin{equation}
  W(Z)
  \coloneqq
  \left[
    \hat\theta_1(Z)-\tfrac{\pi}{4},
    \hat\theta_1(Z)+\tfrac{\pi}{4}
  \right]
  \subset \mathbb S^1.
  \label{eq:window_definition}
\end{equation}
Here $W(Z)$ is interpreted as an estimator for an arc on the circle,
equivalently as an estimator of an interval modulo $2\pi$. The length $\pi/2$ is
the local identifiability scale used for the homodyne model generated by
squeezed-vacuum probes, as discussed in Subsection~\ref{subsec:sensitivity}. The
restriction to $W(Z)$ provides the local parameter range on which the adaptive
homodyne likelihood is treated as identifiable.

Stage~II uses squeezed-vacuum states and adaptive homodyne measurements for
efficient, and asymptotically optimal phase estimation. Let
$X=(X_1,\ldots,X_{N_2})$ denote the measurement outcomes of Stage~II. At the
$\ell$th homodyne adaptive measurement, the local-oscillator phase $\phi_\ell$
is selected as a measurable function of the information available up to that
point,
\begin{equation}
  \phi_\ell
  =
  \Phi_\ell(Z,X_1,\ldots,X_{\ell-1})
  \in[0,2\pi),
  \label{eq:adaptive_policy}
\end{equation}
so that Stage~II is described by an adaptive POVM. Conditional on
$(Z,X_1,\ldots,X_{\ell-1})$, the $\ell$th measurement is the homodyne POVM
$\Pi_{\phi_\ell}(dx)$, corresponding to the quadrature $X_{\phi_\ell}$ defined
in Subsection~\ref{subsub:homodyne}.

For a squeezed-vacuum probe $\ket{\zeta_2}$ with $\zeta_2=r_2e^{i\psi_2}$, the
homodyne outcome is normally distributed with zero mean and variance that
depends on the unknown phase $\theta$ through the relative angle
\begin{equation}
  u \coloneq  \theta+\phi-\frac{\psi_2}{2}
  \pmod{\pi}.
  \label{eq:relative_angle_u}
\end{equation}

More precisely, for a homodyne setting $\phi$,
\begin{equation}
  \begin{aligned}
  \sigma_\theta(\phi)^2
  &=
  \frac12\left(e^{-2r_2}\cos^2 u+e^{2r_2}\sin^2 u\right) \\
  &=
    \frac12\left(\cosh(2r_2)-\sinh(2r_2)\cos(2u)\right).
    \end{aligned}
  \label{eq:svs_homodyne_variance_stage2}
\end{equation}
Since the mean of the distribution of outcomes is identically zero, the Fisher
information of one homodyne sample is the variance term. Substitution of
Eq.~\eqref{eq:svs_homodyne_variance_stage2} into Eq.~\eqref{eq:fi_homodyne}
gives
\begin{equation}
    \begin{aligned}
  \mathcal I_{\mathrm{hom}}(\theta;\phi)
  &=
  \frac{
    \left(\partial_\theta\sigma_\theta(\phi)^2\right)^2
  }{
    2\,\sigma_\theta(\phi)^4
  }\\
  &=
  \frac{
    2\sinh^2(2r_2)\sin^2(2u)
  }{
    \left(\cosh(2r_2)-\sinh(2r_2)\cos(2u)\right)^2
      }.
      \end{aligned}
  \label{eq:fi_homodyne_stage2_closedform}
\end{equation}

For fixed $\theta$ and $r_2$, the information extracted by a homodyne
measurement is determined by the offset $u$ in Eq.~\eqref{eq:relative_angle_u}.
By adapting $\phi_\ell$, the measurement can therefore be kept close to an
offset $u_\ast(r_2)$ that maximizes
Eq.~\eqref{eq:fi_homodyne_stage2_closedform}, up to the symmetries of the model.
At such offsets, the homodyne measurement is locally optimal for squeezed-vacuum
and attains the QFI $\mathcal F_Q^{\mathrm{SVS}}$,
\begin{equation}
\begin{split}
  \max_{\phi\in[0,2\pi)}
  \mathcal I_{\mathrm{hom}}(\theta;\phi)
  &= 8\,\bar n(0,\zeta_2)\left[\bar n(0,\zeta_2)+1\right] \\
  &=  2\sinh^2(2r_2) = \mathcal F_Q^{\mathrm{SVS}}.
\end{split}
\label{eq:homodyne_attains_qfi}
\end{equation}
where $\bar n(0,\zeta_2)=\sinh^2 r_2$ \cite{Monras2006,Pinel2013,Olivares2009}.
Equivalently, one may choose $u_\ast(r_2)$ such that
$\cos(2u_\ast)=\tanh(2r_2)$, up to the symmetries of the homodyne model.

We apply the adaptive scheme from \cite{RodrguezGarca2024} in Stage~II
inside the window $W(Z)$ selected by Stage~I. The policy implements the local
maximization of $\mathcal I_{\mathrm{hom}}$ by replacing the unknown $\theta$
with the current maximum-likelihood estimate, constrained to $W(Z)$, and then
choosing $\phi_\ell$ so that the induced offset $u$ in
Eq.~\eqref{eq:relative_angle_u} remains close to $u_\ast(r_2)$. Although
\cite{RodrguezGarca2024} is formulated for a real squeezing parameter, the same
rule extends to complex $\zeta_2$ by incorporating the squeeze angle $\psi_2$
through Eq.~\eqref{eq:relative_angle_u}.

After $N_2$ homodyne measurements, the adaptive policy in Stage~II defines a
conditional likelihood for the homodyne outcomes on the interval selected by
Stage~I. This likelihood is combined with the Stage~I heterodyne likelihood, and
the final optimization is restricted to $W(Z)$. Let $L_1(\cdot;Z)$ denote the
Stage~I likelihood, and let $L_2(\cdot;X\mid Z,\Phi)$ denote the Stage~II
likelihood induced by the adaptive homodyne policy $\Phi$. The final estimator
is any measurable maximizer of the joint two-stage likelihood restricted to the
interval selected in Stage~I:
\begin{equation}
  \begin{aligned}
    \hat\theta(Z,X) &\in \argmax_{\theta'\in W(Z)} \left\{ \log L_1(\theta';Z)
                      \right. \\ &\hspace{4.5em} + \log L_2\left(\theta';X\mid \left. Z,\Phi
                                   \right) \right \}.
  \end{aligned}
  \label{eq:restricted_ml}
\end{equation}
The restriction to $W(Z)$ makes the Stage~II likelihood identifiable on the
selected interval, while the inclusion of $\log L_1$ retains the localization
information from Stage~I in the final decision rule.

\subsection{Outliers and the overshoot mechanism}
\label{subsec:overshoot_mechanism}

The estimator in Eq.~\eqref{eq:restricted_ml} is $W(Z)$-valued by construction.
That is, for every realization $(z,x)$ of $(Z,X)$, $\hat\theta(z,x)\in W(z).$
This constraint does not force a large error when the Stage~I window contains
the true phase. At finite energy, however, the Stage~I window may fail to
contain the true parameter. If $\theta\notin W(z)$, then the final estimator is
forced to remain in an arc that excludes the true phase. The resulting circular
error is bounded below by the distance from $\theta$ to the selected window,
independently of the information collected in Stage~II. To quantify this effect,
fix $\theta\in[0,2\pi)$ and let $\mathcal Z$ denote the Stage~I outcome space.
The set of Stage~I outcomes for which the selected window contains the true
phase is
\begin{equation}
  \widetilde{\mathcal C}_\theta
  =
  \{z\in\mathcal Z:\theta\in W(z)\}.
\end{equation}
Equivalently, evaluated at the random Stage~I outcome $Z$, the coverage event is
\begin{equation}
  \mathcal C_\theta
  =
  \{Z\in\widetilde{\mathcal C}_\theta\}
  =
  \{\theta\in W(Z)\}.
  \label{eq:covarage_event}
\end{equation}
Then, the notation $\theta\in W(Z)$ denotes the event that the random window
selected from the Stage~I outcome contains the true phase.
Since $W(Z)$ is the arc of half-width $\pi/4$ centered at $\hat\theta_1(Z)$,
this event agrees, up to endpoint conventions, with
\begin{equation}
  \mathcal C_\theta
  =
  \left\{
    d\left(\hat\theta_1(Z),\theta\right)\le \frac{\pi}{4}
  \right\}.
\end{equation}

We measure the failure of coverage by
\begin{equation}
  \Delta_\theta(Z)
  \coloneqq
  \mathrm{dist}_{\mathbb S^1}\!\left(\theta,W(Z)\right)
  =
  \left(
    d\left(\hat\theta_1(Z),\theta\right)
    -
    \frac{\pi}{4}
  \right)_+,
  \label{eq:overshoot_def}
\end{equation}
where $d$ is the geodesic distance on the circle, as defined in
Eq.~\eqref{eq:geodesic_distance}, and $(u)_+\coloneqq\max\{u,0\}$. Thus,
$\Delta_\theta(Z)$ depends only on the Stage~I outcome. It vanishes whenever the
selected window contains the true phase, and it measures how far the true phase
lies outside the selected window when the localization step fails. The following
lemma gives the corresponding pointwise lower bound on the error of any
estimator constrained to $W(Z)$.

\begin{lemma}[Overshoot lower bound]
  \label{lem:overshoot_lb}
  Let $\hat\theta$ be an estimator satisfying $\hat\theta(Z,X)\in W(Z)$ almost
  surely. Then, for every realization $(z,x)$ for which $\hat\theta(z,x)\in
  W(z)$ and every $\theta\in[0,2\pi)$,
  \begin{equation}
    d\left(\hat\theta(z,x),\theta\right)^2
    \ge
    \Delta_\theta(z)^2.
    \label{eq:overshoot_pointwise_lb}
\end{equation}
Consequently,
\begin{equation}
  \mathrm{MSE}\left(\hat\theta;\theta\right)
  =
  \mathbb E_\theta\!\left[
    d\left(\hat\theta(Z,X),\theta\right)^2
  \right]
  \ge
  \mathbb E_\theta\!\left[\Delta_\theta(Z)^2\right].
  \label{eq:overshoot_mse_lb}
\end{equation}
\end{lemma}

\begin{proof}
  The proof has two steps. First, we prove a deterministic geometric inequality
  for a fixed realization $(z,x)$. Then we evaluate this inequality at the
  random outcome $(Z,X)$ and take expectation.

  Fix $\theta\in[0,2\pi)$ and fix a realization $(z,x)$ such that
  $\hat\theta(z,x)\in W(z)$. By definition,
\begin{equation}
  \Delta_\theta(z)
  =
  \operatorname{dist}_{\mathbb S^1}\!\left(\theta,W(z)\right)
  =
  \inf_{\vartheta\in W(z)} d(\vartheta,\theta).
\end{equation}
Since $\hat\theta(z,x)$ is one of the points in $W(z)$, the distance from
$\theta$ to $\hat\theta(z,x)$ cannot be smaller than the infimum of the
distances from $\theta$ to all points in $W(z)$. Hence
\begin{equation}
  d\left(\hat\theta(z,x),\theta\right)
  \ge
  \inf_{\vartheta\in W(z)} d(\vartheta,\theta)
  =
  \Delta_\theta(z).
\end{equation}
Both sides are nonnegative, so squaring gives
Eq.~\eqref{eq:overshoot_pointwise_lb}.

Now evaluate this pointwise inequality at the random realization $(Z,X)$. Since
$\hat\theta(Z,X)\in W(Z)$ almost surely, we obtain
\begin{equation}
  d\left(\hat\theta(Z,X),\theta\right)^2
  \ge
  \Delta_\theta(Z)^2
  \qquad
  \text{almost surely under } \mathbb P_\theta .
\end{equation}
Taking expectation with respect to the joint distribution of $(Z,X)$ under the
parameter value $\theta$ gives Eq.~\eqref{eq:overshoot_mse_lb}.
\end{proof}

Lemma~\ref{lem:overshoot_lb} explains why a purely local benchmark, such as the
QCRB, is not sufficient to characterize the performance of a full-period
two-stage protocol. On the coverage event $\mathcal C_\theta$, Stage~II may be
locally efficient; however, on $\mathcal C_\theta^c$ the estimator is
constrained to an interval that does not contain the true phase. The risk
therefore contains a contribution determined by the Stage~I localization error,
$\mathbb E_\theta\!\left[\Delta_\theta(Z)^2\right],$ which depends on both the
probability of missing the correct window and the squared distance from the true
phase to the selected window. This contribution is not controlled by the local
Fisher information of Stage~II. It must therefore appear explicitly in any
finite-energy lower bound for the full-period estimation problem. We use this
observation to derive a generalized Cramér-Rao bound that separates the local
contribution of Stage~II from the Stage~I overshoot penalty.

\subsection{Generalized Cramér-Rao bound for two-stage estimators}
\label{subsec:generalized_crb}

We decompose the risk of the two-stage estimator into two contributions: the
error from the local Stage~II estimation on the coverage event, and the error
caused by Stage~I localization failures on the complement. In general, the
ordinary Cramér-Rao inequality applies to regular statistical models in which
the parameter is locally identified \cite{Lehman1998}. In the setting considered
here, such a local model is available only after conditioning on Stage~I
outcomes $Z=z$ for which the true phase lies within the selected window.

Fix $\theta_0\in[0,2\pi)$, and let $\pi_{\mathbb S^1}:\mathbb R\to\mathbb S^1$
denote the map that sends a real number $t$ to its equivalence class on the
circle, identifying $t$ and $t+2\pi k$ for every $k\in\mathbb Z$. We use
$\tilde\theta_0$ for the representative of $\theta_0$ in $[0,2\pi)$. For an
outcome $z$ from Stage~I such that $\theta_0\in W(z)$, let $I_z\subset\mathbb R$
be the connected component of $\pi_{\mathbb S^1}^{-1}(W(z))$ that contains
$\tilde\theta_0$. Since $W(z)$ has length $\pi/2$, the restriction of
$\pi_{\mathbb S^1}$ to $I_z$ is a one-to-one parametrization of the selected
arc. Thus each $\vartheta\in W(z)$ has a unique lift $\tilde\vartheta\in I_z$,
and
\begin{equation}
  d(\vartheta,\theta_0)
  =
  |\tilde\vartheta-\tilde\theta_0|,
  \qquad
  \vartheta\in W(z).
\end{equation}
In particular, because the final estimator is constrained to $W(z)$, the
quantity $d(\hat\theta(z,X),\theta_0)$ agrees with the absolute error of the
lifted estimator in the coordinate interval $I_z$.

For such a fixed $z$, the adaptive policy $\Phi$ in Stage~II induces a
conditional statistical model for the corresponding outcomes. We denote
$p_{2,z}(x;\vartheta)$ the density of this conditional model, with respect to a
dominating measure $\mu_z$, when the lifted phase coordinate is $\vartheta\in
I_z$. Equivalently, the physical phase is $\pi_{\mathbb S^1}(\vartheta)$. The
conditional Fisher information of the Stage~II model at $\theta_0$ is then
\begin{equation}
  \mathcal I_{2}(\theta_0;z)
  \coloneqq
  \mathbb E_{\theta_0}\!\left[
    \left(
      \partial_\vartheta \log p_{2,z}(X;\vartheta)
      \big|_{\vartheta=\tilde\theta_0}
    \right)^2
    \Bigm| Z=z
  \right].
  \label{eq:conditional_stage2_fi}
\end{equation}
The notation $\mathcal I_2(\theta_0;z)$ keeps the original circular parameter
value in the argument, but the derivative is taken in the lifted coordinate
$I_z$. On outcomes $z$ for which $\theta_0\notin W(z)$, this local coordinate is
not used; in the theorem below, the corresponding term is always multiplied by
the indicator of the coverage event.

The conditional Fisher information $\mathcal I_2(\theta_0;z)$ enters the bound
through the local Cramér-Rao inequality. This inequality is applied only when
Stage~I has selected an interval containing the true phase; outside this event,
the proof does not use a local Cramér-Rao argument and instead relies on the
geometric overshoot bound of Lemma~\ref{lem:overshoot_lb}. Therefore, we do not
impose a global regularity assumption on the full-circle estimation problem.
Instead, after conditioning on a successful Stage~I localization and lifting the
selected interval to a local coordinate, we impose the local hypotheses needed
to apply the Cramér-Rao inequality to the conditional Stage~II model.

To state this precisely, fix $\theta_0\in[0,2\pi)$ and let $\mathbb
P_{\theta_0}^Z$ denote the marginal distribution of the Stage~I outcome $Z$
under the true value of the parameter. We now consider the coverage set
$\mathcal C_{\theta_0}$ in Eq.~\eqref{eq:covarage_event}. For each $z\in\mathcal
C_{\theta_0}$, the selected interval $W(z)$ admits a lifted coordinate interval
$I_z\subset\mathbb R$. We denote by $\tilde\theta_0(z)\in I_z$ the corresponding
lift of the true phase and by $\widetilde{\hat\theta}_z(x)\in I_z$ the unique
lift of the final estimate $\hat\theta(z,x)\in W(z)$. Expectations with respect
to the conditional Stage~II model $p_{2,z}(x;\vartheta)$ are denoted by $\mathbb
E_{\vartheta,z}$. With this notation, we impose the following local regularity
assumptions.

\begin{assumptions}[Local regularity on the coverage event]
  \label{ass:local_crb}
  For $\mathbb P_{\theta_0}^Z$-almost every $z\in\mathcal C_{\theta_0}$, we
  assume the following:
\begin{enumerate}
\item The conditional Stage~II model
  $\{p_{2,z}(x;\vartheta):\vartheta\in I_z\}$ is regular in a neighborhood of
  $\tilde\theta_0(z)$. In particular, it is dominated by a common measure, its
  density is differentiable in $\vartheta$, and the usual differentiations
  under the integral sign are valid.
\item The lifted estimator $\widetilde{\hat\theta}_z$ is locally unbiased at
  $\tilde\theta_0(z)$.
\item The conditional Fisher information satisfies
  $0<\mathcal I_2(\theta_0;z)<\infty$.
\end{enumerate}
\end{assumptions}

With these local regularity assumptions in place, we can state the generalized
two-stage Cramér-Rao bound. The result separates the contribution from
successful Stage~I localization, where a local Cramér-Rao inequality applies,
from the contribution of localization failures, which is controlled by the
overshoot penalty. The proof is given in Appendix~\ref{app:proof_main_theorem}.

\begin{theorem}[Generalized two-stage Cram\'er--Rao bound]
\label{thm:generalized_crb}
Suppose the true parameter value is $\theta_0\in[0,2\pi)$, with corresponding
Stage~I coverage set $\mathcal C_{\theta_0}$. Let $\hat\theta(Z,X)$ be a
two-stage estimator constrained to the selected Stage~I window, so that
$\hat\theta(Z,X)\in W(Z)$ $\mathbb P_{\theta_0}$-almost surely. Suppose that the
Local regularity assumptions, Assumptions~\ref{ass:local_crb}, are satisfied on
this coverage set. Then
\begin{equation}
  \mathrm{MSE}\left(\hat\theta;\theta_0\right)
  \ge
  \mathbb E_{\theta_0}\!\left[
    \frac{
      \mathbf 1_{\mathcal C_{\theta_0}}(Z)
    }{
      \mathcal I_2(\theta_0;Z)
    }
  \right]
  +
  \mathbb E_{\theta_0}\!\left[
    \Delta_{\theta_0}(Z)^2
  \right].
  \label{eq:generalized_crb_main}
\end{equation}
where the quotient in the first term is understood to be zero when
$Z\notin\mathcal C_{\theta_0}$, $\mathcal I_2(\theta_0;Z)$ is the conditional
Stage~II Fisher information associated with the Stage~I outcome $Z$, and
$\Delta_{\theta_0}(Z)$ is the overshoot distance from the true phase to the
selected Stage~I window.
\end{theorem}

Theorem~\ref{thm:generalized_crb} is expressed in terms of the conditional
classical Fisher information $\mathcal I_2(\theta_0;Z)$ generated by the
Stage~II measurement. A measurement-independent version follows by upper
bounding this quantity with the quantum Fisher information of the Stage~II probe
family. For the fixed Stage~II product family $\rho_\theta^{(2)\otimes N_2}$, we
denote this total quantum Fisher information at $\theta_0$ by $\mathcal
F_{Q,2}^{\mathrm{tot}} = \mathcal F_Q^{\rho_\theta^{(2)\otimes
    N_2}}(\theta_0).$ The following corollary gives the corresponding
quantum version of the two-stage bound.

\begin{corollary}[Generalized Quantum CRB for two-Stage protocols]
\label{cor:generalized_crb_qfi}
Suppose the true parameter value is $\theta_0\in[0,2\pi)$, with corresponding
Stage~I coverage set $\mathcal C_{\theta_0}$. Let $\hat\theta(Z,X)$ be a
two-stage estimator constrained to the selected Stage~I window, so that
$\hat\theta(Z,X)\in W(Z)$ $\mathbb P_{\theta_0}$-almost surely. Suppose that the
Local regularity assumptions, Assumptions~\ref{ass:local_crb}, are satisfied on
this coverage set.

Assume, in addition, that for each Stage~I outcome $z$, the Stage~II adaptive
policy defines a $\theta$-independent POVM on the fixed Stage~II product probe
family, and that
$0<\mathcal F_{Q,2}^{\mathrm{tot}}<\infty$. Then, for
$\mathbb P_{\theta_0}^Z$-almost every $z\in\mathcal C_{\theta_0}$,
\begin{equation}
  \mathcal I_2(\theta_0;z)
  \le
  \mathcal F_{Q,2}^{\mathrm{tot}},
  \label{eq:conditional_bc_bound}
\end{equation}
and consequently,
\begin{equation}
  \mathrm{MSE}\left(\hat\theta;\theta_0\right)
  \ge
  \frac{
    \mathbb P_{\theta_0}\!\left(Z\in\mathcal C_{\theta_0}\right)
  }{
    \mathcal F_{Q,2}^{\mathrm{tot}}
  }
  +
  \mathbb E_{\theta_0}\!\left[
    \Delta_{\theta_0}(Z)^2
  \right].
  \label{eq:generalized_crb_qfi_version}
\end{equation}
\end{corollary}
\begin{proof}
Fix a Stage~I outcome $z\in\mathcal C_{\theta_0}$ for which
Assumptions~\ref{ass:local_crb} hold. Conditional on $Z=z$, the Stage~II
adaptive policy defines a POVM on the fixed Stage~II product probe family. This
POVM may depend on the realized Stage~I outcome $z$, and its homodyne settings
may depend on previous Stage~II outcomes, but the resulting adaptive
measurement is still independent of the unknown parameter $\theta$. Hence, by
the Braunstein--Caves inequality \cite{Braunstein1994}, the conditional
classical Fisher information is bounded by the Stage~II quantum Fisher
information:
\begin{equation}
  \mathcal I_2(\theta_0;z)
  \le
  \mathcal F_{Q,2}^{\mathrm{tot}} .
  \label{eq:proof_bc_bound}
\end{equation}
This holds for $\mathbb P_{\theta_0}^Z$-almost every
$z\in\mathcal C_{\theta_0}$. Since
$0<\mathcal I_2(\theta_0;z)$ and
$0<\mathcal F_{Q,2}^{\mathrm{tot}}<\infty$, Eq.~\eqref{eq:proof_bc_bound}
implies the indicator-weighted bound
\begin{equation}
  \frac{
    \mathbf 1_{\mathcal C_{\theta_0}}(Z)
  }{
    \mathcal I_2(\theta_0;Z)
  }
  \ge
  \frac{
    \mathbf 1_{\mathcal C_{\theta_0}}(Z)
  }{
    \mathcal F_{Q,2}^{\mathrm{tot}}
  },
  \label{eq:proof_reciprocal_qfi_bound}
\end{equation}
with the same convention as in Theorem~\ref{thm:generalized_crb}: the quotient
on the left is taken to be zero when $Z\notin\mathcal C_{\theta_0}$.

Substituting Eq.~\eqref{eq:proof_reciprocal_qfi_bound} into the first term of
Eq.~\eqref{eq:generalized_crb_main} gives
\begin{equation}
\begin{split}
  \mathbb E_{\theta_0}\!\left[
    \frac{
      \mathbf 1_{\mathcal C_{\theta_0}}(Z)
    }{
      \mathcal I_2(\theta_0;Z)
    }
  \right]
  &\ge
  \mathbb E_{\theta_0}\!\left[
    \frac{
      \mathbf 1_{\mathcal C_{\theta_0}}(Z)
    }{
      \mathcal F_{Q,2}^{\mathrm{tot}}
    }
  \right]  \\
  &=
  \frac{
    \mathbb P_{\theta_0}\!\left(Z\in\mathcal C_{\theta_0}\right)
  }{
    \mathcal F_{Q,2}^{\mathrm{tot}}
  } .
\end{split}
\label{eq:proof_qfi_contribution}
\end{equation}
Combining this bound with Theorem~\ref{thm:generalized_crb} proves
Eq.~\eqref{eq:generalized_crb_qfi_version}.
\end{proof}

Theorem~\ref{thm:generalized_crb} applies to any two-stage protocol whose final
estimator is constrained to the window selected in Stage~I and whose conditional
Stage~II model satisfies the local assumptions on the coverage event. The first
term in Eq.~\eqref{eq:generalized_crb_main} is the local Cramér-Rao
contribution, averaged over the event that Stage~I obtains a window containing
the true phase. The second term is the overshoot penalty caused by localization
failures in Stage~I. Thus local efficiency in Stage~II can improve the risk only
on the coverage event; when the selected window misses the true phase, the
estimator incurs a geometric error that is independent of the local Fisher
information available in Stage~II.

In summary, the generalized bound of Eq.~\eqref{eq:generalized_crb_main}
captures the relevant errors from the class of two-Stage estimation protocols
considered here. Its first term is the local information-limited contribution,
now explicitly restricted to the event on which the local problem in Stage~II is
well posed. Its second term is a purely geometric tail penalty determined by
Stage~I and is insensitive the performance of the measurement implemented in
Stage~II. In the subsequent sections we analyze the conditional information term
Stage~II and evaluate the overshoot penalty $\mathbb
E_\theta[\Delta_\theta(Z)^2]$ for different families of quantum states in
Stage~I. This procedure yields tractable finite-energy lower bounds. In
addition, it provides guidance for how to best allocate the available energy
between the coarse localization in Stage~I and locally optimal phase estimation
in Stage~II to minimize the overall error.

\section{Coarse localization in Stage~I with coherent-states}
\label{sec:coherent_init}

As a first step, we investigate the performance of coarse localization in
Stage~I using different Gaussian probes: coherent states, which serve as a
tractable classical benchmark; and displaced squeezed states, which can provide
an improved performance. We explicitly compute the coverage probability and
overshoot penalty from Stage~I that enter the lower bound of
Theorem~\ref{thm:generalized_crb}. In this section, we focus on coherent state
probes, and Section \ref{sec:dsvs_init} discusses the use of displaced squeezed
vacuum states in Stage~I.

\subsection{Stage~I model and the heterodyne maximum-likelihood estimator}

Let the Stage~I input probe be the coherent state $\ket{\alpha_1}$, with
$\alpha_1=|\alpha_1|e^{i\phi_1}$ and $|\alpha_1|>0$. Under the phase-shift
operation $U_\theta=e^{-i\theta\hat n}$, the encoded state is
$\ket{\alpha_{1,\theta}}$, where $\alpha_{1,\theta}=\alpha_1e^{-i\theta}
=|\alpha_1|e^{i(\phi_1-\theta)}$. Heterodyne measurement on each copy produces
independent outcomes $\beta_k\in\mathbb C$, $k=1,\dots,N_1$, with
$\beta_k\sim\mathcal{CN}(\alpha_{1,\theta},1)$. Here, $\mathcal{CN}(\mu,1)$
denotes the circular complex normal distribution with mean $\mu$ and unit
variance.

The log-likelihood for $\theta$ based on $Z=(\beta_1,\dots,\beta_{N_1})$ is
\begin{equation} \ell_1(\theta;Z) =
-\sum_{k=1}^{N_1}\left|\beta_k-\alpha_{1,\theta}\right|^2 + \mathrm{const}.
  \label{eq:coherent_stage1_loglik}
\end{equation} Expanding the first term of
Eq.~\eqref{eq:coherent_stage1_loglik}, we have
\begin{equation} \ell_1(\theta;Z) =
-\sum_{k=1}^{N_1}\left|\beta_k-\bar\beta\right|^2 -
N_1\left|\bar\beta-\alpha_{1,\theta}\right|^2 + \mathrm{const},
\end{equation} where $\bar\beta \coloneqq \frac{1}{N_1}\sum_{k=1}^{N_1}\beta_k$
is the sample mean. The first term in the above equation is independent of
$\theta$, so maximizing $\ell_1(\theta;Z)$ is equivalent to minimizing the
Euclidean distance $\left|\bar\beta-\alpha_{1,\theta}\right|$ between
$\alpha_{1,\theta}$ and the sample mean. Consequently, the maximum-likelihood
estimator (MLE) in Stage~I
becomes
\begin{equation} \hat\theta_1(Z) = \phi_1 -
\Arg\!\left(\sum_{k=1}^{N_1}\beta_k\right) = \phi_1 - \Arg(B_{N_1}).
  \label{eq:coherent_stage1_mle}
\end{equation}

We obtain the signed wrapped error in Stage~I as $\varepsilon_1 =
\operatorname{wrap}(\hat\theta_1-\theta) =
\Arg\!\left(e^{i(\hat\theta_1-\theta)}\right).$ For convenience, we introduce
the effective energy $E_{\mathrm{coh}}=N_1|\alpha_1|^2$ for the family of
$N_{1}$ coherent states. With this definition, and by rotational invariance, the
distribution of $\varepsilon_1$ depends on $(N_1,\alpha_1)$ only through
$E_{\mathrm{coh}}$.

\subsection{Probability distribution of the phase estimator in Stage~I}

The circular distance for the estimator in Eq.~\eqref{eq:coherent_stage1_mle} is
$d(\hat\theta_1,\theta)=|\varepsilon_1|$. By rotational invariance, the
distribution of $\varepsilon_1$ is independent of both the true phase $\theta$
and the probe phase $\phi_1$. Its dependence on $N_1$ and $\alpha_1$ enters only
through the effective coherent energy $E_{\mathrm{coh}}$. The resulting density
is the following.

\begin{proposition}[Rician density of the Stage~I phase error]
  \label{prop:rician_phase_law} For any Stage~I model, with coherent probes and
  heterodyne measurement, the signed wrapped error $\varepsilon_1$ has density
  $f_{\varepsilon}(\varepsilon;E_{\mathrm{coh}})$ given, for
  $\varepsilon\in(-\pi,\pi]$, by
\begin{equation}
\begin{split}
  f_{\varepsilon}(\varepsilon;E_{\mathrm{coh}})
  &=
    \frac{e^{-E_{\mathrm{coh}}}}{2\pi}
    +
    \frac{\sqrt{E_{\mathrm{coh}}}\cos\varepsilon}{2\sqrt{\pi}}\,
    e^{-E_{\mathrm{coh}}\sin^2\varepsilon}  \\
  &\quad\times
    \left[
    1+\erf\!\left(\sqrt{E_{\mathrm{coh}}}\cos\varepsilon\right)
    \right].
\end{split}
\label{eq:rician_phase_density_coherent}
\end{equation}
\end{proposition}

\begin{proof}[Proof sketch]
The likelihood in Eq.~\eqref{eq:coherent_stage1_loglik} depends on the
heterodyne record only through the sum of the outcomes. It is convenient to use
the normalized statistic
\begin{equation}
  B_{N_1}=\frac{1}{\sqrt{N_1}}\sum_{k=1}^{N_1}\beta_k .
\end{equation}
Since the outcomes are independent circular complex normal random variables,
\begin{equation}
  B_{N_1}
  \sim
  \mathcal{CN}\!\left(
    \sqrt{E_{\mathrm{coh}}}\,e^{i(\phi_1-\theta)},1
  \right).
\end{equation}
The normalization does not change the argument, so
$\hat\theta_1=\phi_1-\Arg(B_{N_1})$. Therefore,
\begin{equation}
  \varepsilon_1
  =
  \Arg\!\left(e^{i(\hat\theta_1-\theta)}\right)
  =
  \Arg\!\left(
    e^{i(\phi_1-\theta)}B_{N_1}^*
  \right).
\end{equation}
The rotated conjugate $e^{i(\phi_1-\theta)}B_{N_1}^*$ is a circular complex
normal random variable with real positive mean $\sqrt{E_{\mathrm{coh}}}$ and
unit variance. Thus $\varepsilon_1$ is the phase of a nonzero-mean circular
complex Gaussian random variable. Writing this variable in Cartesian
coordinates, changing to polar coordinates, and integrating out the radial
coordinate gives the density in Eq.~\eqref{eq:rician_phase_density_coherent}. A
complete derivation is given in Appendix~\ref{app:proof_rician_phase_law}.
\end{proof}

Using the identity $1+\erf(x)=2\Phi(\sqrt{2}x)$, where $\Phi$ denotes the
standard normal cumulative distribution function,
Eq.~\eqref{eq:rician_phase_density_coherent} can be written equivalently as
\begin{equation}
\begin{split}
f_{\varepsilon}(\varepsilon;E_{\mathrm{coh}})
&= \frac{e^{-E_{\mathrm{coh}}}}{2\pi}
+
\frac{\sqrt{E_{\mathrm{coh}}}\cos\varepsilon}{\sqrt{\pi}},
e^{-E_{\mathrm{coh}}\sin^2\varepsilon}  \\
& \times
\Phi\left(\sqrt{2E_{\mathrm{coh}}}\cos\varepsilon\right).
\end{split}
\label{eq:rician_phase_density_coherent_phi}
\end{equation}
This is the Rician phase density, i.e., the phase distribution of a nonzero-mean
circular complex Gaussian random variable \cite{Rice1945}. Since the Stage~I
localization window has length $\pi/2$, its half-width is $\Delta=\pi/4$.
Therefore, Eq.~\eqref{eq:rician_phase_density_coherent_phi} yields explicit
one-dimensional expressions for the Stage~I coverage probability and the
corresponding overshoot penalty.

Since the Stage~I window has length $\pi/2$, its half-width is $\delta=\pi/4$.
Using the density $f_{\varepsilon}(\varepsilon;E_{\mathrm{coh}})$ from
Proposition~\ref{prop:rician_phase_law}, we introduce the coverage function
\begin{equation}
  C_{\mathrm{coh}}(E_{\mathrm{coh}})
  \coloneqq
  \int_{-\delta}^{\delta}
  f_{\varepsilon}(\varepsilon;E_{\mathrm{coh}})\,d\varepsilon
  \label{eq:coherent_coverage_exact}
\end{equation}
and the overshoot function
\begin{equation}
\begin{split}
  G_{\mathrm{coh}}(E_{\mathrm{coh}})
  \coloneqq
  \int_{|\varepsilon|>\delta}
  &\left(|\varepsilon|-\delta\right)^2  \\
  &\times
  f_{\varepsilon}(\varepsilon;E_{\mathrm{coh}})\,d\varepsilon .
\end{split}
  \label{eq:coherent_overshoot_exact}
\end{equation}
In Eq.~\eqref{eq:coherent_overshoot_exact}, the integration domain is the
complement of $[-\delta,\delta]$ in $(-\pi,\pi]$.

These two functions are precisely the Stage~I localization quantities that enter
the generalized two-stage bound. The first is the probability that the selected
window contains the true phase, while the second is the unconditional overshoot
contribution associated with localization failures. This identification is
stated explicitly in the following corollary.

\begin{corollary}[Coverage and overshoot for coherent probes in Stage~I]
\label{cor:coherent_coverage_overshoot}
For Stage~I with coherent probes and heterodyne measurement, the probability
that the selected window contains the true phase is
\begin{equation}
  \mathbb P_\theta\!\left(Z\in\mathcal C_\theta\right)
  =
  C_{\mathrm{coh}}(E_{\mathrm{coh}}),
\end{equation}
and the corresponding overshoot penalty is
\begin{equation}
  \mathbb E_\theta\!\left[
    \Delta_\theta(Z)^2
  \right]
  =
  G_{\mathrm{coh}}(E_{\mathrm{coh}}).
\end{equation}
Both quantities are independent of the true phase $\theta$ and depend on the
Stage~I coherent probes only through the effective energy
$E_{\mathrm{coh}}$.
\end{corollary}

Consequently, when Stage~I uses coherent probes, the overshoot term in
Theorem~\ref{thm:generalized_crb} becomes an explicit function of the Stage~I
coherent energy. For every two-stage estimator satisfying the hypotheses of
Theorem~\ref{thm:generalized_crb}, we obtain
\begin{equation}
  \mathrm{MSE}\left(\hat\theta;\theta_0\right)
  \ge
  \mathbb E_{\theta_0}\!\left[
    \frac{
      \mathbf 1_{\mathcal C_{\theta_0}}(Z)
    }{
      \mathcal I_2(\theta_0;Z)
    }
  \right]
  +
  G_{\mathrm{coh}}(E_{\mathrm{coh}}).
  \label{eq:coherent_explicit_crb_conditional}
\end{equation}
The first term is the contribution of the local Stage~II estimation problem,
restricted to the event that Stage~I selects a window containing the true phase.
The second term is fully determined by the coherent energy allocated to Stage~I.

Using the quantum version of the bound, Corollary~\ref{cor:generalized_crb_qfi},
gives a measurement-independent version. If $\mathcal F_{Q,2}^{\mathrm{tot}}$
denotes the total quantum Fisher information of the Stage~II squeezed-vacuum
product family, then
\begin{equation}
  \mathrm{MSE}\left(\hat\theta;\theta_0\right)
  \ge
  \frac{
    C_{\mathrm{coh}}(E_{\mathrm{coh}})
  }{
    \mathcal F_{Q,2}^{\mathrm{tot}}
  }
  +
  G_{\mathrm{coh}}(E_{\mathrm{coh}}).
  \label{eq:coherent_explicit_crb_qfi}
\end{equation}

Equation~\eqref{eq:coherent_explicit_crb_qfi} provides an analytic benchmark for
using coherent probes in Stage~I. Increasing $E_{\mathrm{coh}}$ increases the
coverage probability $C_{\mathrm{coh}}(E_{\mathrm{coh}})$ and suppresses the
overshoot penalty $G_{\mathrm{coh}}(E_{\mathrm{coh}})$, making the localization
stage more reliable. Under a fixed total energy constraint, however, allocating
more energy to Stage~I leaves less energy for the squeezed-vacuum probes in
Stage~II. The two-stage protocol therefore exhibits an energy-allocation
tradeoff: reliable coarse localization in Stage~I competes with high local
precision in Stage~II.

\subsection{Stage~I sample complexity}

$C_{\mathrm{coh}}$ and $G_{\mathrm{coh}}$ provide a complete description for
coarse localization with coherent states in Stage~I with finite energy. In
particular, $C_{\mathrm{coh}}$ determines the probability that the selected
window in Stage~I of width $\pi/2$ contains the true phase, while
$G_{\mathrm{coh}}$ quantifies the corresponding overshoot penalty. Since the
distribution of the signed wrapped error $\varepsilon_1$ depends on the coarse
localization in Stage~I only through $E_{\mathrm{coh}}=N_1|\alpha_1|^2$, then we
observe that the number of probes, the single-probe amplitude, and the total
energy in Stage~I are interchangeable resources in Stage~I.

Let $c\in(0,1)$ be the target coverage probability for the Stage~I window. We
define the minimum energy in Stage~I needed to reach this target $c$ as
\begin{equation}
  E_{\mathrm{coh}}^{\min}(c)
  =
  \inf\{E'\ge 0:\ C_{\mathrm{coh}}(E')\ge c\}.
  \label{eq:exact_energy_threshold}
\end{equation}
Equivalently, in terms of the effective coherent-state amplitude
$\rho=\sqrt{E_{\mathrm{coh}}}=\sqrt{N_1}\,|\alpha_1|$, we define
\begin{equation}
  \rho_{\min}(c)
  =
  \inf\{\rho\ge 0:\ C_{\mathrm{coh}}(\rho^2)\ge c\}.
  \label{eq:exact_amplitude_threshold}
\end{equation}
For a fixed single-probe amplitude $|\alpha_1|$, the corresponding minimum
number of probes in Stage~I becomes
\begin{equation}
  N_1^{\min}(c)
  =
  \left\lceil
    \frac{E_{\mathrm{coh}}^{\min}(c)}{|\alpha_1|^2}
  \right\rceil.
  \label{eq:exact_sample_complexity}
\end{equation}
Thus $E_{\mathrm{coh}}^{\min}(c)$, $\rho_{\min}(c)$, and
$N_1^{\min}(c)$ are equivalent ways of describing the localization cost needed
to achieve coverage level $c$.

The function $c\mapsto \rho_{\min}(c)$ is defined implicitly through the
coverage function $C_{\mathrm{coh}}$. For numerical calculations, it is helpful
to replace it by a closed-form empirical approximation. We use the two-parameter
probit model
\begin{equation}
  \rho_{\min}(c)
  \approx
  a\,\Phi^{-1}\!\bigl(1-b(1-c)\bigr),
  \label{eq:coherent_probit_fit}
\end{equation}
where $a,b>0$ are fitted by regression over the confidence range of interest. We
note that this approximation is not an asymptotic statement, but rather an
empirical fit to obtain a threshold curve in the finite-energy regime.

A separate analytic approximation is obtained in the large-energy regime. The
Fisher information of a single heterodyne measurement on a coherent state of
amplitude $|\alpha_1|$ is $2|\alpha_1|^2$. Hence, for $N_1$ independent probes,
$\mathcal I_{N_1} = 2N_1|\alpha_1|^2 = 2E_{\mathrm{coh}}.$ The MLE is
asymptotically efficient, and, in the local
asymptotic regime, the signed wrapped error satisfies\
\begin{equation}
  \sqrt{E_{\mathrm{coh}}}\,\varepsilon_1
  \xrightarrow{d}
  \mathcal N\!\left(0,\frac12\right).
  \label{eq:coherent_asymptotic_normality}
\end{equation}
Equivalently, for large $E_{\mathrm{coh}}$, $\varepsilon_1 \approx \mathcal
N\!\left(0,\frac{1}{2E_{\mathrm{coh}}}\right).$ Approximating the coverage
probability by this normal distribution gives
\begin{equation}
  \mathbb P_\theta\!\left(|\varepsilon_1|\le \frac{\pi}{4}\right)
  \approx
  2\Phi\!\left(
    \frac{\pi}{4}\sqrt{2E_{\mathrm{coh}}}
  \right)-1.
\end{equation}
Solving for the target coverage $c$ yields
\begin{equation}
  E_{\mathrm{coh}}^{\min}(c)
  \approx
  \frac{8}{\pi^2}\,z_{(1+c)/2}^2,
  \label{eq:coherent_energy_threshold_asymptotic}
\end{equation}
and therefore
\begin{equation}
  \rho_{\min}(c)
  \approx
  \frac{2\sqrt{2}}{\pi}\,z_{(1+c)/2},
  \label{eq:coherent_amplitude_threshold_asymptotic}
\end{equation}
where $z_{(1+c)/2}$ is the $(1+c)/2$ quantile of the standard normal
distribution. These formulas describe the high-energy normal regime, but they do
not capture the finite-energy corrections contained in the curve $c\mapsto
\rho_{\min}(c)$.

Figure~\ref{fig:coherent_threshold_curve} compares three descriptions of the map
$c\mapsto \rho_{\min}(c)$. The blue points represent the exact curve defined
implicitly by Eq.~\eqref{eq:exact_amplitude_threshold} and obtained by
numerically by inverting $C_{\mathrm{coh}}(\rho^2)$. The orange solid curve is
the probit approximation of Eq.~\eqref{eq:coherent_probit_fit}, and the red
dashed curve is the asymptotic approximation from
Eq.~\eqref{eq:coherent_amplitude_threshold_asymptotic}. For any fixed
single-probe amplitude $|\alpha_1|$, this threshold curve immediately yields the
minimum number of probes required in Stage~I through
Eq.~\eqref{eq:exact_sample_complexity}.

\begin{figure}[h]
  \centering
  \includegraphics[width=\columnwidth]{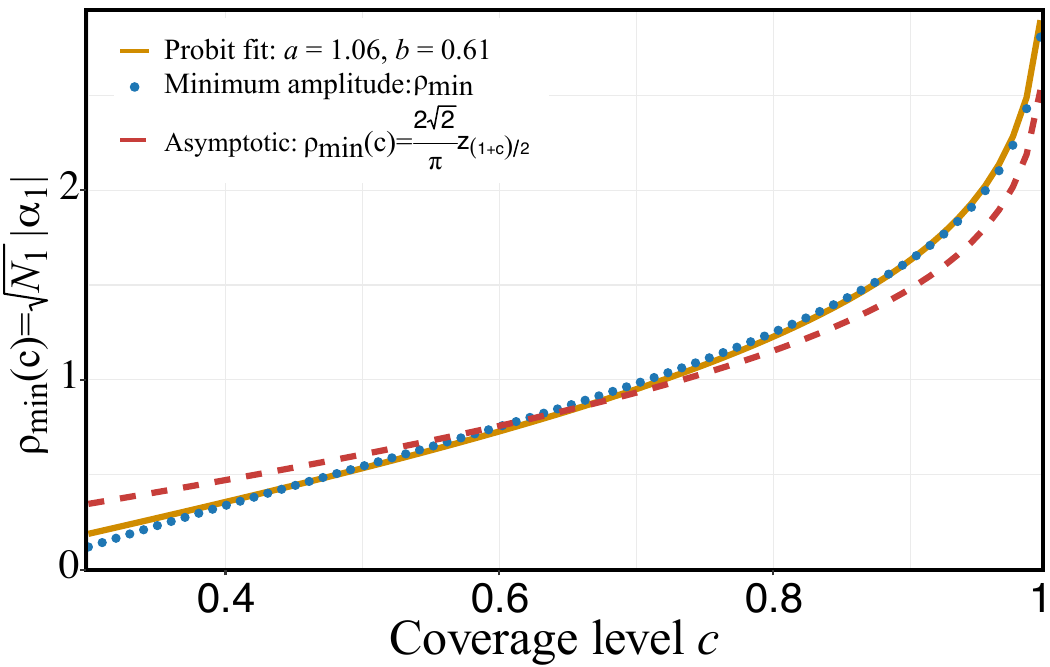}
  \caption{\label{fig:coherent_threshold_curve} Minimum effective coherent
    amplitude $\rho_{\min}(c)$ required for the Stage~I window of width $\pi/2$
    to attain coverage probability at least $c$. The points represent the exact
    threshold curve obtained by numerically inverting
    Eq.~\eqref{eq:exact_amplitude_threshold}. The solid curve is the
    finite-energy probit approximation in Eq.~\eqref{eq:coherent_probit_fit},
    and the dashed curve is the asymptotic approximation in
    Eq.~\eqref{eq:coherent_amplitude_threshold_asymptotic}, derived from the
    normal approximation in Eq.~\eqref{eq:coherent_asymptotic_normality}. The
    discrepancy between the exact and asymptotic curves at moderate confidence
    levels is the result of finite-energy corrections to the normal regime.}
\end{figure}

\section{Coarse localization in Stage~I with displaced squeezed probes}
\label{sec:dsvs_init}

This section analyzes the performance of coarse localization in Stage~I with
displaced squeezed probes and compares this performance with the benchmark from
coherent-states in Sec. \ref{sec:coherent_init}.

\subsection{Local information gain under heterodyne measurements}
\label{subsec:dsvs_fisher}

Consider single-mode displaced squeezed probes $\ket{\alpha,\zeta}$, with
$\alpha=|\alpha|e^{i\phi}\in\mathbb C$, and $\zeta=re^{i\psi}\in\mathbb C$.
Under a phase shift operation, the family of parameterized states becomes
$\ket{\alpha,\zeta;\theta}=U_\theta\ket{\alpha,\zeta}$, with energy $E=\bar
n(\alpha,\zeta)=|\alpha|^2+\sinh^2 r$. The heterodyne measurement applied on to
this family produces a Gaussian model has a mean $m_\theta=d _\theta$ and
covariance $\Sigma_\theta = V_\theta+\tfrac12\mathbb I_2$ that both depend on
$\theta$, see Sec.~\ref{sec:preliminaries}. Specifically, the unitary
transformation $U_\theta\ket{\alpha,\zeta}$ transform the displacement and
squeezing phases as $\phi\mapsto\phi-\theta$ and $\psi\mapsto\psi-2\theta$. In
the real quadrature representation, the first-moment vector of the state is
\begin{equation}
  d_\theta
  =
  \sqrt{2}|\alpha|
  \begin{pmatrix}
    \cos(\phi-\theta) \\
    \sin(\phi-\theta)
  \end{pmatrix},
  \label{eq:dtheta_dsvs}
\end{equation}
and its covariance matrix $V_\theta$ has entries
\begin{equation}
\begin{aligned}
  [V_\theta]_{11}
  &=
  \frac12\!\left[
    \cosh(2r)-\sinh(2r)\cos(\psi-2\theta)
  \right], \\
  [V_\theta]_{22}
  &=
  \frac12\!\left[
    \cosh(2r)+\sinh(2r)\cos(\psi-2\theta)
  \right], \\
  [V_\theta]_{12}
  &=
  [V_\theta]_{21}
  =
  -\frac12\sinh(2r)\sin(\psi-2\theta).
\end{aligned}
\label{eq:Vtheta_dsvs}
\end{equation}
Thus, the heterodyne outcome distribution is Gaussian with mean $d_\theta$ and
covariance $V_\theta+\frac12\mathbb I_2$. For convenience, we introduce the
relative phase $\chi \coloneqq 2\phi-\psi$, which quantifies the alignment
between the displacement direction and the principal axes of the covariance
matrix of the squeezed state.

The performance of displaced squeezed probes in Stage~I is characterized by the
single-probe quantum Fisher information and by the classical Fisher information
generated by heterodyne measurement. Eq.~\eqref{eq:fi_heterodyne} gives an
explicit expression for the Fisher information with heterodyne measurements.

\begin{proposition}[Fisher information of displaced squeezed probes]
\label{prop:dsvs_fisher}
For the phase-encoded family generated by $U_\theta=e^{-i\theta\hat n}$ from the
displaced squeezed state $\ket{\alpha,\zeta}$ with relative phase $\chi$ and
squeezing parameter $r >0$, the single-probe quantum Fisher information is
\begin{equation}
\begin{split}
  \mathcal F_Q^{\mathrm{DSVS}}(\alpha,\zeta)
  &=
  2\sinh^2(2r)  \\
  &\quad+
  4|\alpha|^2
  \left[
    \cosh(2r)-\sinh(2r)\cos\chi
  \right],
\end{split}
\label{eq:qfi_dsvs_full_rewrite}
\end{equation}
For heterodyne measurement on the same phase-encoded family, the corresponding
single-probe classical Fisher information is
\begin{equation}
  \mathcal I_{\mathrm{het}}^{\mathrm{DSVS}}(\alpha,\zeta)
  =
  4\sinh^2 r
  +
  2|\alpha|^2
  \left(
    1-\tanh r\,\cos\chi
  \right).
  \label{eq:cfi_dsvs_full_rewrite}
\end{equation}
\end{proposition}

\begin{proof}
The proof evaluates the two Fisher informations by using the natural generator
and measurement model for phase estimation. The quantum Fisher information is
obtained from the photon-number variance of the input pure state, while the
heterodyne Fisher information follows by substituting the first moment and
covariance matrix of the phase-encoded Gaussian state into
Eq.~\eqref{eq:fi_heterodyne}.

For the unitary family generated by $U_\theta=e^{-i\theta \hat n}$, the
pure-state quantum Fisher information is
\begin{equation}
  \mathcal F_Q^{\mathrm{DSVS}}(\alpha,\zeta)
  =
  4\,\mathrm{Var}_{\ket{\alpha,\zeta}}(\hat n).
\end{equation}
For a displaced squeezed state
$\ket{\alpha,\zeta}=D(\alpha)S(\zeta)\ket{0}$, with
$\alpha=|\alpha|e^{i\phi}$ and $\zeta=re^{i\psi}$, a standard Gaussian moment
calculation gives
\begin{equation}
\begin{split}
  \mathrm{Var}_{\ket{\alpha,\zeta}}(\hat n)
  &=
  \frac12\sinh^2(2r)  \\
  &\quad+
  |\alpha|^2
  \left[
    \cosh(2r)-\sinh(2r)\cos\chi
  \right],
\end{split}
\end{equation}
where $\chi=2\phi-\psi$. Therefore,
\begin{equation}
\begin{split}
  \mathcal F_Q^{\mathrm{DSVS}}(\alpha,\zeta)
  &=
  2\sinh^2(2r)  \\
  &\quad+
  4|\alpha|^2
  \left[
    \cosh(2r)-\sinh(2r)\cos\chi
  \right],
\end{split}
\end{equation}
which proves Eq.~\eqref{eq:qfi_dsvs_full_rewrite}.

We now compute the heterodyne Fisher information. Under the phase shift, the
displacement phase and squeezing phase transform as
$\phi\mapsto\phi-\theta$ and $\psi\mapsto\psi-2\theta$, so the relative phase
$\chi=2\phi-\psi$ is invariant. Hence the Fisher information is independent of
the value of $\theta$, and we may evaluate Eq.~\eqref{eq:fi_heterodyne} at
$\theta=0$.

The contribution from the parameter dependence of the first moment is
\begin{equation}
\begin{split}
  &\left.
  (\partial_\theta d_\theta)^\top
  \!\left(V_\theta+\tfrac12\mathbb I_2\right)^{-1}
  (\partial_\theta d_\theta)
  \right|_{\theta=0}  \\
  &\qquad =
  2|\alpha|^2
  \left(
    1-\tanh r\,\cos\chi
  \right),
\end{split}
\end{equation}
while the contribution from the parameter dependence of the covariance matrix is
\begin{equation}
\begin{split}
  &\left.
  \frac12
  \Tr\!\left[
    \left\{
      \left(V_\theta+\tfrac12\mathbb I_2\right)^{-1}
      \partial_\theta V_\theta
    \right\}^2
  \right]
  \right|_{\theta=0}  \\
  &\qquad =
  4\sinh^2 r .
\end{split}
\end{equation}
Substituting these two terms into Eq.~\eqref{eq:fi_heterodyne} gives
\begin{equation}
  \mathcal I_{\mathrm{het}}^{\mathrm{DSVS}}(\alpha,\zeta)
  =
  4\sinh^2 r
  +
  2|\alpha|^2
  \left(
    1-\tanh r\,\cos\chi
  \right),
\end{equation}
which proves Eq.~\eqref{eq:cfi_dsvs_full_rewrite}.
\end{proof}

Proposition~\ref{prop:dsvs_fisher} shows that displaced squeezed probes can
increase the local information available in Stage~I relative to coherent states.
In the coherent-state limit $r=0$, the classical Fisher information reduces to
$\mathcal I_{\mathrm{het}}=2|\alpha|^2$ and the QFI to $\mathcal
F_Q^{\mathrm{DSVS}}(\alpha,\zeta)=4|\alpha|^2$. However, for $r>0$, the
information is not fixed by the total energy alone. This dependence appears
through the factors $\cosh(2r)-\sinh(2r)\cos\chi$ in the QFI and $1-\tanh
r\,\cos\chi$ in the Fisher information of heterodyne measurement. Thus, for
fixed $|\alpha|$ and $r$, the displacement contribution is largest when
$\cos\chi=-1$ and smallest when $\cos\chi=1$.

It is useful to express these quantities at fixed Stage~I probe energy. Since
$E=|\alpha|^2+\sinh^2 r,$ the admissible squeezing range is $0\le r\le
\operatorname{arsinh}\sqrt{E}$, with $|\alpha|^2=E-\sinh^2 r$. In this
parametrization, the Fisher information of heterodyne measurement becomes
\begin{equation}
\begin{split}
  \mathcal I_{\mathrm{het}}^{(\chi)}(r;E)
  &=
  4\sinh^2 r  \\
  &\quad+
  2\left(E-\sinh^2 r\right)
  \left(1-\tanh r\,\cos\chi\right),
\end{split}
\label{eq:cfi_dsvs_fixed_energy}
\end{equation}
whereas the QFI is
\begin{equation}
\begin{split}
  \mathcal F_Q^{(\chi)}(r;E)
  &=
  2\sinh^2(2r)  \\
  &\quad+
  4\left(E-\sinh^2 r\right) \\
  &\qquad\times
  \left[
    \cosh(2r)-\sinh(2r)\cos\chi
  \right].
\end{split}
\label{eq:qfi_dsvs_fixed_energy}
\end{equation}

The endpoint $r=0$ corresponds to coherent states as probes. In this case all
the Stage~I probe energy is carried by the displacement, and $\mathcal
I_{\mathrm{het}}^{(\chi)}(0;E)=2E,$ and $\mathcal F_Q^{(\chi)}(0;E)=4E_1.$ At
the opposite endpoint, $r=\operatorname{arsinh}\sqrt{E}$ and $|\alpha|=0$, all
the energy is carried by squeezing. This gives a squeezed-vacuum probe, for
which Fisher information is $\mathcal I_{\mathrm{het}}^{(\chi)}=4E,$ and the
QFI is $\mathcal F_Q^{(\chi)}=8E(E+1),$ with the relative phase $\chi$
becoming irrelevant because the displacement vanishes. This squeezed-vacuum
endpoint is useful as a local benchmark, but it is not suitable for coarse
localization over the full circle in Stage~I. Indeed, when $|\alpha|=0$, the
phase-encoded family has a $\pi$-rotation symmetry (see
Sec.~\ref{subsec:sensitivity}), so phases separated by $\pi$ cannot be
distinguished. Thus, displaced squeezed probes exhibit a tradeoff between the
local gain obtained from squeezing and the global identifiability provided by a
nonzero displacement.

Figure~\ref{fig:cfi_dsvs} shows the (local) gain for displaced squeezed probes
for a fixed energy $E=\sinh^2(1.4)$: the classical Fisher information with
heterodyne in Fig.~\ref{fig:cfi_dsvs}(a), and the QFI in
Fig.~\ref{fig:cfi_dsvs}(b). The value $r=1.4$ corresponds to a squeezing of
$\approx12.2\,\mathrm{dB}$, which is experimentally realistic for optical
squeezing \cite{vahlbruch2016detection}. Note that varying $r$ redistributes
energy in Stage~I between displacement and squeezing. The curves for $\chi=\pi$
and $\chi=0$ correspond to the two extremal phase matching conditions, with
$\chi=\pi$ maximizing and $\chi=0$ minimizing the local information.
Specifically, for $\chi=\pi$ displaced squeezed states show an advantage over
the coherent-state benchmark over a broad range of $r$, while retaining a
nonzero displacement and hence full-period identifiability. On the other hand,
for $\chi=0$, the information can instead fall below the coherent-state
benchmark, even for $r\neq0$. Thus, squeezing alone is not sufficient to produce
a local advantage, and the relative phase must also be chosen appropriately.
\begin{figure}[t]
  \centering

  \begin{minipage}{0.48\textwidth}
    \centering
    \includegraphics[width=\textwidth]{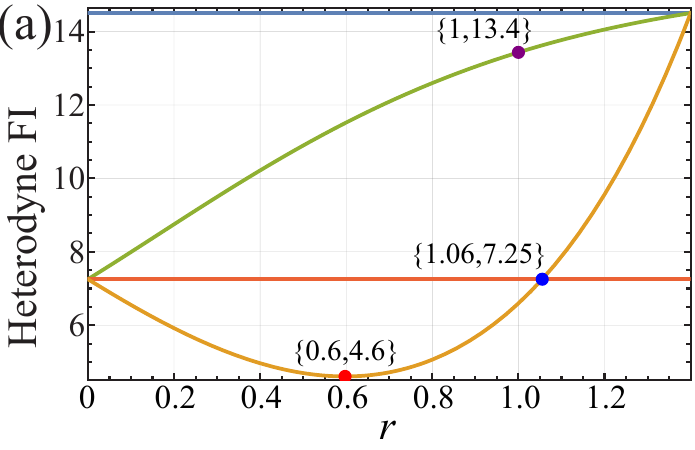}

  \end{minipage}
  \hfill
  \begin{minipage}{0.48\textwidth}
    \centering
    \includegraphics[width=\textwidth]{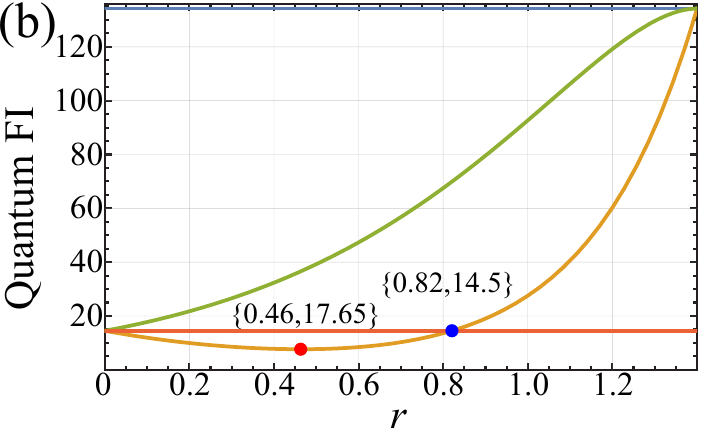}

  \end{minipage}

  \caption{\label{fig:cfi_dsvs} Local information for displaced squeezed probes
    under the fixed Stage~I energy budget $E=\sinh^2(1.4)$. Panel~(a) shows
    the heterodyne Fisher information $\mathcal
    I_{\mathrm{het}}^{(\chi)}(r;E)$, and panel~(b) shows the quantum Fisher
    information $\mathcal F_Q^{(\chi)}(r;E)$. As $r$ varies, the total energy
    is fixed while the allocation between displacement and squeezing changes.
    The red and blue horizontal curves mark the coherent endpoint $r=0$ and the
    squeezed-vacuum endpoint $|\alpha|=0$, respectively. The curves $\chi=\pi$
    and $\chi=0$ show constructive and destructive phase matching.}
\end{figure}

We note that these local bounds quantify the information in the neighborhood of
a fixed phase. They do not imply, by themselves, improved finite-sample
localization on $[0,2\pi)$. Although the condition $|\alpha|>0$ allows for phase
identifiability, it does not guarantee that the MLE in Stage~I has smaller error
compared to the coherent-state benchmark. We therefore realize the finite-sample
analysis of the Stage~I error distribution in the next subsection.

\subsection{Likelihood function for heterodyne measurement and MLE error
  distribution}
\label{subsec:dsvs_phase_law}

%We analyze the probability distribution of the localization error in Stage~I for displaced
%squeezed probes.
In contrast to the case with coherent states described in Section
\ref{sec:coherent_init}, displaced squeezed probes yield a heterodyne outcome
distribution with covariance that depends on the unknown phase. The likelihood
therefore depends not only on the sample mean, but also on second-order
information in the data. Consequently, for $N_1>1$ there is in general no
reduction to a single effective heterodyne observation of the same form as in
the coherent-state case. To analyze this problem, we first treat the one-probe
case $N_1=1$, where the geometry of the likelihood gives an integral
representation for the error density.

Consider a displaced squeezed probe $\ket{\alpha_1,\zeta_1}$, with
$\alpha_1=|\alpha_1|e^{i\phi_1}$ and $\zeta_1=r_1e^{i\psi_1}$. Let $Z\in\mathbb
R^2$ be the heterodyne measurement outcome. As recalled in
Section~\ref{sec:preliminaries}, $Z$ is Gaussian with mean $d_\theta$ and
covariance $V_\theta+\frac12\mathbb I_2$. To describe the error distribution, we
rotate the outcome by the true displacement direction. Thus, for the purpose of
the analysis, set $Y = \mathsf R\!\left(-(\theta+\phi_1)\right)Z.$ Then $Y$ is
Gaussian with mean
\begin{equation}
  \mu_1 e_1,
  \qquad
  \mu_1=\sqrt{2}\,|\alpha_1|,
  \qquad
  e_1=(1,0)^\top,
\end{equation}
and covariance
\begin{equation}
  \Sigma_{\delta_1}
  =
  \mathsf R(\delta_1)
  \begin{pmatrix}
    \lambda_{1,-} & 0\\
    0 & \lambda_{1,+}
  \end{pmatrix}
  \mathsf R(\delta_1)^\top .
  \label{eq:dsvs_aligned_covariance}
\end{equation}
Here $\delta_1=\frac{\psi_1}{2}-\phi_1=-\frac{\chi_1}{2},$ with
$\chi_1=2\phi_1-\psi_1,$. Then, $\delta_1$ is the angle between the displacement
direction and the principal axes of the covariance ellipse, whose variances are
\begin{equation}
  \lambda_{1,\pm}
  =
  \frac{1+e^{\pm 2r_1}}{2}.
  \label{eq:dsvs_heterodyne_eigenvalues}
\end{equation}

We then can write the aligned outcome in polar coordinates as
$Y=S_Y\,v(\Gamma),$ $v(\gamma)=(\cos\gamma,\sin\gamma)^\top,$ $S_Y\ge 0,$ where
$\Gamma\in(-\pi,\pi].$ Here, $S_Y=\|Y\|$ is the radial coordinate of the
observed heterodyne measurement outcome in the aligned frame of reference (note
that it should not be confused with the squeezing strength $r_1$). The joint
density of $(S_Y,\Gamma)$ with respect to $ds\,d\gamma$ is
\begin{equation}
\begin{split}
  f_{S_Y,\Gamma}(s,\gamma)
  &=
  \frac{s}{2\pi\sqrt{\det\Sigma_{\delta_1}}}
  \exp\!\left[
    -\frac12
    \left(s\,v(\gamma)-\mu_1e_1\right)^\top
    \Sigma_{\delta_1}^{-1} \right.  \\
  &\hspace{5.4em}\times
  \left.  \left(s\,v(\gamma)-\mu_1e_1\right)
  \right].
\end{split}
\label{eq:dsvs_joint_polar_density}
\end{equation}
for $s\ge 0$ and $\gamma\in(-\pi,\pi]$.

%The geometry of this procedure is shown in Figure~\ref{fig:dsvs_geometry}.
When the likelihood is evaluated at a candidate phase $\vartheta$, the aligned
observation is rotated by the phase error $\vartheta-\theta$. Thus, the
likelihood depends on $\vartheta$ through $\eta=\Gamma-(\vartheta-\theta).$ For
a fixed observed radius $S_Y=s$, the maximization over the candidate phase is
equivalent to maximizing the likelihood over the angular mismatch
$\eta\in(-\pi,\pi]$. The restriction of the log-likelihood to the circle of
radius $s$, up to an additive constant, is
\begin{equation}
  \Lambda_s(\eta)
  =
  -\frac12
  \left(s\,v(\eta)-\mu_1e_1\right)^\top
  \Sigma_{\delta_1}^{-1}
  \left(s\,v(\eta)-\mu_1e_1\right).
  \label{eq:dsvs_radial_profile}
\end{equation}
Let $\eta_\ast(s)$ be a measurable maximizing branch,
\begin{equation}
  \eta_\ast(s)\in
  \arg\max_{\eta\in(-\pi,\pi]} \Lambda_s(\eta).
  \label{eq:dsvs_optimal_offset_branch}
\end{equation}
Then, the signed wrapped error of the one-probe MLE is
\begin{equation}
  \varepsilon_1
  =
  \Arg\!\left(e^{i(\hat\theta_1-\theta)}\right)
  =
  \Arg\!\left(e^{i(\Gamma-\eta_\ast(S_Y))}\right),
  %\qquad
  %\varepsilon_1\in(-\pi,\pi].
  \label{eq:dsvs_mle_error_from_branch}
\end{equation}
If $\Gamma-\eta_\ast(S_Y)$ is already taken on the principal branch, this is
equivalently written as $\varepsilon_1=\Gamma-\eta_\ast(S_Y)$ , with the
corresponding circular distance
$d\left(\hat\theta_1,\theta\right)=|\varepsilon_1|.$ Thus, unlike in the case
with coherent states in Section \ref{sec:coherent_init}, the angular correction
is not constant. It depends on the observed radius $S_Y$. This dependence is a
finite-sample effect of the anisotropy of the heterodyne likelihood.

The following proposition formalizes this geometric description as the density
of the one-probe signed error, which is required for calculating the coverage
and overshoot terms in Stage~I for displaced squeezed probes.

\begin{proposition}[One-probe error density for displaced squeezed probes and
  heterodyne measurements]
\label{prop:dsvs_error_integral_representation}
The signed wrapped error $\varepsilon_1 =
\Arg\!\left(e^{i(\hat\theta_1-\theta)}\right) \in(-\pi,\pi]$ has probability
density
\begin{equation}
  f_{\varepsilon}^{\mathrm{DSVS}}(\varepsilon)
  =
  \int_0^\infty
  f_{S_Y,\Gamma}\!\left(
    s,\,
    \varepsilon+\eta_\ast(s)
  \right)\,ds,
  %\qquad
  %\varepsilon\in(-\pi,\pi],
  \label{eq:dsvs_error_density_integral}
\end{equation}
where the angular argument is understood modulo $2\pi$ and represented in
$(-\pi,\pi]$.
\end{proposition}
The proof is given in Appendix~\ref{app:dsvs_error_density}. 

In the coherent-state limit $r_1=0$, one has $\lambda_{1,-}=\lambda_{1,+}=1$,
and hence $\Sigma_{\delta_1}=\mathbb I_2$. The restricted likelihood
Eq.~\eqref{eq:dsvs_radial_profile} is then maximized at $\eta_\ast(s)=0$ for
every $s>0$; at $s=0$, any maximizing branch is equivalent, so we choose
$\eta_\ast(0)=0$. Therefore, Eq.~\eqref{eq:dsvs_error_density_integral} reduces
to the Rician phase density in Proposition~\ref{prop:rician_phase_law}, with
$E_{\mathrm{coh}}=|\alpha_1|^2$. The details of this reduction are given in
Appendix~\ref{app:dsvs_coherent_limit}. The density in
Proposition~\ref{prop:dsvs_error_integral_representation} determines the
one-probe analogues of the Stage~I coverage and overshoot terms. Since the
Stage~I window has length $\pi/2$, we write its half-width as $\delta=\pi/4$.
The one-probe coverage probability is
\begin{equation}
  C_{\mathrm{DSVS}}^{(1)} =
  \mathbb P_\theta\!\left(|\varepsilon_1|\le \delta\right) =
  \int_{-\delta}^{\delta}
  f_{\varepsilon}^{\mathrm{DSVS}}(\varepsilon)\,d\varepsilon,
  \label{eq:dsvs_exact_coverage}
\end{equation}
and the one-probe overshoot contribution is
\begin{equation}
  \begin{split}
    G_{\mathrm{DSVS}}^{(1)} &= \mathbb E_\theta\!\left[
                              \left(|\varepsilon_1|-\delta\right)_+^2 \right] \\ &=
                                                                                   \int_{|\varepsilon|>\delta} \left(|\varepsilon|-\delta\right)^2
                                                                                   f_{\varepsilon}^{\mathrm{DSVS}}(\varepsilon)\,d\varepsilon.
                                                            \end{split}
                                                            \label{eq:dsvs_exact_overshoot}
                                                          \end{equation}
                                                          These quantities are
                                                          the one-probe
                                                          displaced-squeezed
                                                          counterparts of the
                                                          coherent-probe
                                                          coverage and overshoot
                                                          functions entering the
                                                          generalized two-stage
                                                          bound.

                                                          For $N_1>1$, the same
                                                          construction defines a
                                                          probability measure on
                                                          the signed error
                                                          space, but no closed
                                                          one-dimensional
                                                          density is available
                                                          in general. Let
                                                          $Z=(Z_1,\dots,Z_{N_1})\in(\mathbb
                                                          R^2)^{N_1}$ denote the
                                                          heterodyne measurement
                                                          outcomes, and let
                                                          $\ell_{N_1}(\vartheta;Z)$
                                                          be the sample
                                                          log-likelihood. After
                                                          choosing a measurable
                                                          maximizing branch, the
                                                          MLE in Stage~I is a
                                                          map $T_{N_1}:(\mathbb
                                                          R^2)^{N_1}\to[0,2\pi),
                                                          $ defined by
\begin{equation}
  T_{N_1}(Z)\in
  \arg\max_{\vartheta\in[0,2\pi)}
  \ell_{N_1}(\vartheta;Z).
  \label{eq:dsvs_multishot_mle_map}
\end{equation}
The signed wrapped error in Stage~I is then the measurable map
\begin{equation}
  \varepsilon_1(Z)
  =
  \Arg\!\left(e^{i(T_{N_1}(Z)-\theta)}\right)
  \in(-\pi,\pi].
  \label{eq:dsvs_multishot_signed_error}
\end{equation}
The corresponding circular distance is
$d\left(T_{N_1}(Z),\theta\right)=|\varepsilon_1(Z)|.$ If $P_\theta^{(N_1)}$
denotes the joint probability measure of the sample $Z$, then the probability
distribution of $\varepsilon_1$ is the pushforward measure
$(\varepsilon_1)_\#P_\theta^{(N_1)}$ \cite{JimenezNaszodiVilla2014}. Thus, for
every Borel set $B\subset(-\pi,\pi]$,
\begin{multline}
  \mathbb{P}_\theta(\varepsilon_1\in B)
  =
  P_\theta^{(N_1)}
  \!\left(
    \left\{
      z\in(\mathbb{R}^2)^{N_1}:
    \right.\right.
    \\
    \left.\left.
      \Arg\!\left(e^{i(T_{N_1}(z)-\theta)}\right)\in B
    \right\}
  \right).
  \label{eq:dsvs_pushforward_error_measure}
\end{multline}
If $P_\theta^{(N_1)}$ has density $p_\theta^{(N_1)}$ with respect to Lebesgue
measure on $(\mathbb R^2)^{N_1}$, then equivalently
\begin{equation}
  \mathbb P_\theta(\varepsilon_1\in B)
  =
  \int_{(\mathbb R^2)^{N_1}}
  \mathbf 1_{\left\{
    \Arg(e^{i(T_{N_1}(z)-\theta)})\in B
  \right\}}
  p_\theta^{(N_1)}(z)\,dz.
  \label{eq:dsvs_pushforward_error_density_integral}
\end{equation}
Consequently, when Stage~I uses $N_1>1$ displaced squeezed probes, the Stage~I
coverage and overshoot terms are still computed from the signed MLE error, with
$C_{\mathrm{DSVS}}^{(N_1)} =\mathbb P_\theta(|\varepsilon_1|\le \pi/4)$ and
$G_{\mathrm{DSVS}}^{(N_1)} =\mathbb E_\theta[(|\varepsilon_1|-\pi/4)_+^2]$. In
this case, however, the probability distribution of $\varepsilon_1$ is
determined by the full $N_1$-sample MLE map $T_{N_1}$ in
Eq.~\eqref{eq:dsvs_multishot_mle_map}, applied to the heterodyne sample
$Z=(Z_1,\ldots,Z_{N_1})$. These quantities must therefore be evaluated from the
full heterodyne likelihood, rather than from a closed one-dimensional error
density as in the one-probe case.

\subsection{Sample and energy complexity for displaced squeezed probes in Stage~I}
\label{subsec:dsvs_sample_complexity}

We investigate if displaced squeezing can reduce the resources required for the
coarse localization in Stage~I within an identifiable window of width $\pi/2$
relative to coherent-state probes with heterodyne detection. For this study, we
note that there are two notions of resource. The first is the number of probe
states in Stage~I required to reach a targeted coverage probability. We refer to
this as the Stage~I sample complexity. The second is the total energy spent in
Stage~I. These notions need not give the same ordering, because a displaced
squeezed probe has energy $|\alpha_1|^2+\sinh^2 r_1,$ whereas a coherent-state
with the same displacement amplitude has energy $|\alpha_1|^2$.

Throughout this subsection we take the phase matching to be $\chi_1=\pi$. By
Proposition~\ref{prop:dsvs_fisher}, this choice maximizes the heterodyne Fisher
information for fixed $(|\alpha_1|,r_1)$. The comparison below therefore gives
the most favorable local heterodyne benchmark for displaced squeezed probes in
Stage~I. It does not assert optimality over all possible measurements in
Stage~I; rather, it isolates the effect of changing the family of probe states
under heterodyne measurements.

As a first step, we consider the one-probe regime. For $N_1=1$, the integral
representation in Proposition~\ref{prop:dsvs_error_integral_representation}
gives the probability that the selected window from Stage~I contains the true
phase:
\begin{equation}
  C_{\mathrm{DSVS}}^{(1)}(|\alpha_1|,r_1)
  =
  \mathbb P_\theta\!\left(|\varepsilon_1|\le \frac{\pi}{4}\right)
  =
  \int_{-\pi/4}^{\pi/4}
  f_{\varepsilon}^{\mathrm{DSVS}}(\varepsilon)\,d\varepsilon .
  \label{eq:dsvs_one_shot_coverage}
\end{equation}
Here $\varepsilon_1$ denotes the signed wrapped Stage~I error, and
$|\varepsilon_1|$ is the corresponding circular distance. For a target coverage
level $c\in(0,1)$, define
\begin{equation}
  |\alpha_1|_{\min}^{(1)}(c;r_1)
  =
  \inf\left\{
    a\ge 0:
    C_{\mathrm{DSVS}}^{(1)}(a,r_1)\ge c
  \right\},
  \label{eq:dsvs_one_shot_alpha_threshold}
\end{equation}
and
\begin{equation}
  E_{1,\mathrm{DSVS}}^{\min,(1)}(c;r_1)
  =
  \bigl(|\alpha_1|_{\min}^{(1)}(c;r_1)\bigr)^2+\sinh^2 r_1 .
  \label{eq:dsvs_one_shot_energy_threshold}
\end{equation}
Thus $E_{1,\mathrm{DSVS}}^{\min,(1)}(c;r_1)$ is the minimum energy of a single
displaced squeezed probe that attains a coverage level at least $c$. For
coherent-state probes, the corresponding one-probe threshold is
\begin{equation}
  E_{1,\mathrm{coh}}^{\min,(1)}(c)
  =
  \rho_{\min}(c)^2,
  \label{eq:coh_one_shot_energy_threshold}
\end{equation}
where $\rho_{\min}(c) = \inf\{\rho\ge 0:\ C_{\mathrm{coh}}(\rho^2)\ge c\}$ is
the threshold for coherent states introduced in Section~\ref{sec:coherent_init}.

Figure~\ref{fig:dsvs_stage1_thresholds}(a) shows the resulting energy
thresholds for single probes. In the numerical range considered, we observe that displaced
squeezed states do not reduce the one-probe energy required for coarse
localization, as the corresponding thresholds lie above the curve
for coherent states. The local gain in Fisher information does not automatically translate into a
finite-energy gain for global localization. We conclude then that in one shot,
the estimator must resolve the phase over a nonlocal window, and the additional
curvature of the likelihood is not sufficient to compensate for the energy spent
in squeezing $\sinh^2 r_1$.

We next consider repeated observations. For $N_1>1$, the MLE in Stage~I is computed
from the full heterodyne sample. For fixed $(|\alpha_1|,r_1)$, let
$\varepsilon_1 = \Arg\!\left(e^{i(\hat\theta_1-\theta)}\right) \in(-\pi,\pi]$ be
the signed wrapped error of the $N_1$-sample MLE. We define
\begin{equation}
  C_{\mathrm{DSVS}}^{(N_1)}(|\alpha_1|,r_1)
  =
  \mathbb P_\theta\!\left(
    |\varepsilon_1|\le \frac{\pi}{4}
  \right),
  \label{eq:dsvs_multishot_coverage}
\end{equation}
where the probability is taken under the $N_1$-sample heterodyne model. The
Stage~I sample complexity, as defined by number of
probes used, is
\begin{equation}
  N_{1,\mathrm{DSVS}}^{\min}(c;r_1,|\alpha_1|)
  =
  \min\left\{
    N_1\in\mathbb N:
    C_{\mathrm{DSVS}}^{(N_1)}(|\alpha_1|,r_1)\ge c
  \right\}.
  \label{eq:dsvs_multishot_sample_complexity}
\end{equation}
For
coherent-state probes with the same single-probe amplitude $|\alpha_1|$, the
corresponding benchmark is
\begin{equation}
  N_{1,\mathrm{coh}}^{\min}(c;|\alpha_1|)
  =
  \left\lceil
  \frac{\rho_{\min}(c)^2}{|\alpha_1|^2}
  \right\rceil .
  \label{eq:coh_stage1_sample_complexity}
\end{equation}

Figure~\ref{fig:dsvs_stage1_thresholds}(b) shows the finite-sample
complexity at fixed $|\alpha_1|=0.35$. In contrast to the one-probe energy
comparison, displaced squeezed states can reduce the number of probes required
to attain a prescribed coverage level. The reduction is most visible for
moderate squeezing. This behavior is consistent with the local Fisher
information calculation, since repeated observations concentrate the likelihood
near the true phase. However, the effect is not monotone in $r_1$. Increasing
the squeezing strength beyond a moderate range does not necesarily reduce the
required number of probes. Thus displaced squeezing can reduce the number of
Stage~I observations, but only in an intermediate range of squeezing strengths.
\begin{figure}[t]
  \centering
  \begin{minipage}{0.49\textwidth}
    \centering
    \includegraphics[width=\textwidth]{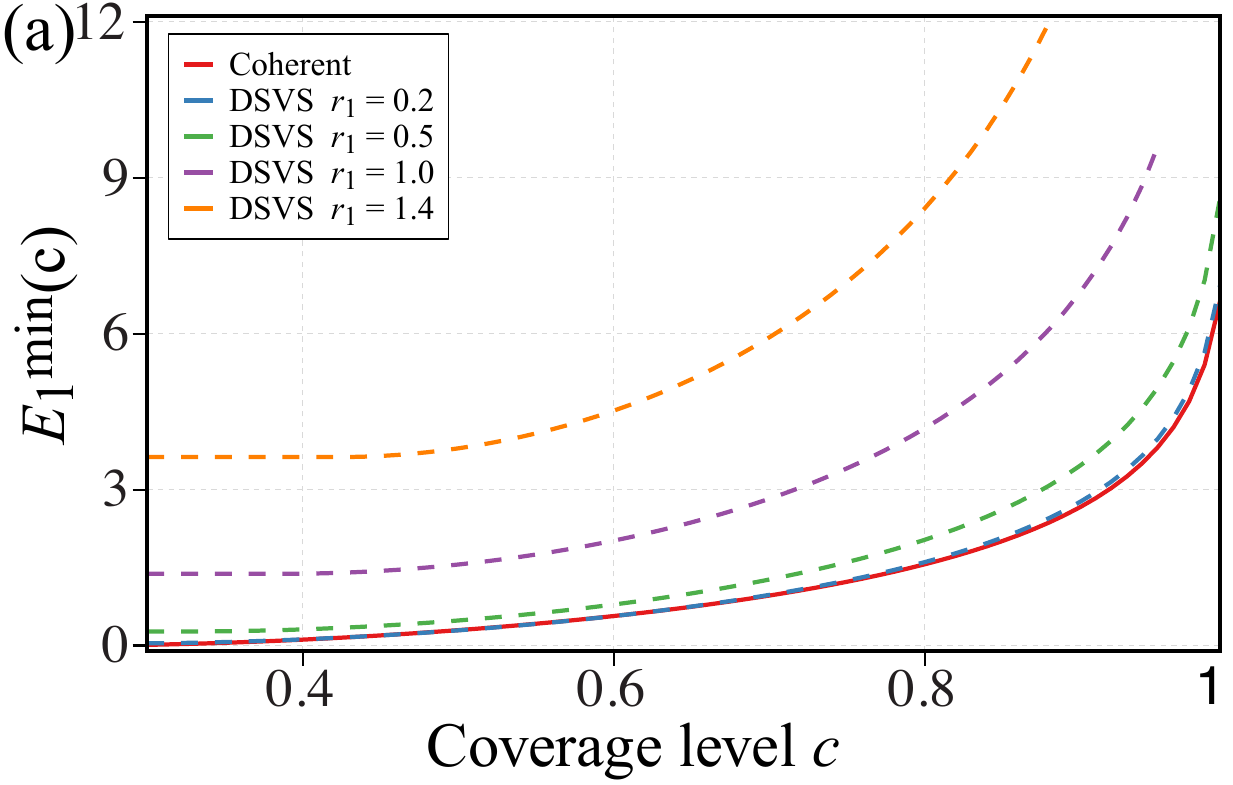}
    \vspace{-0.5em}
  \end{minipage}
  \hfill
  \begin{minipage}{0.49\textwidth}
    \centering
    \includegraphics[width=\textwidth]{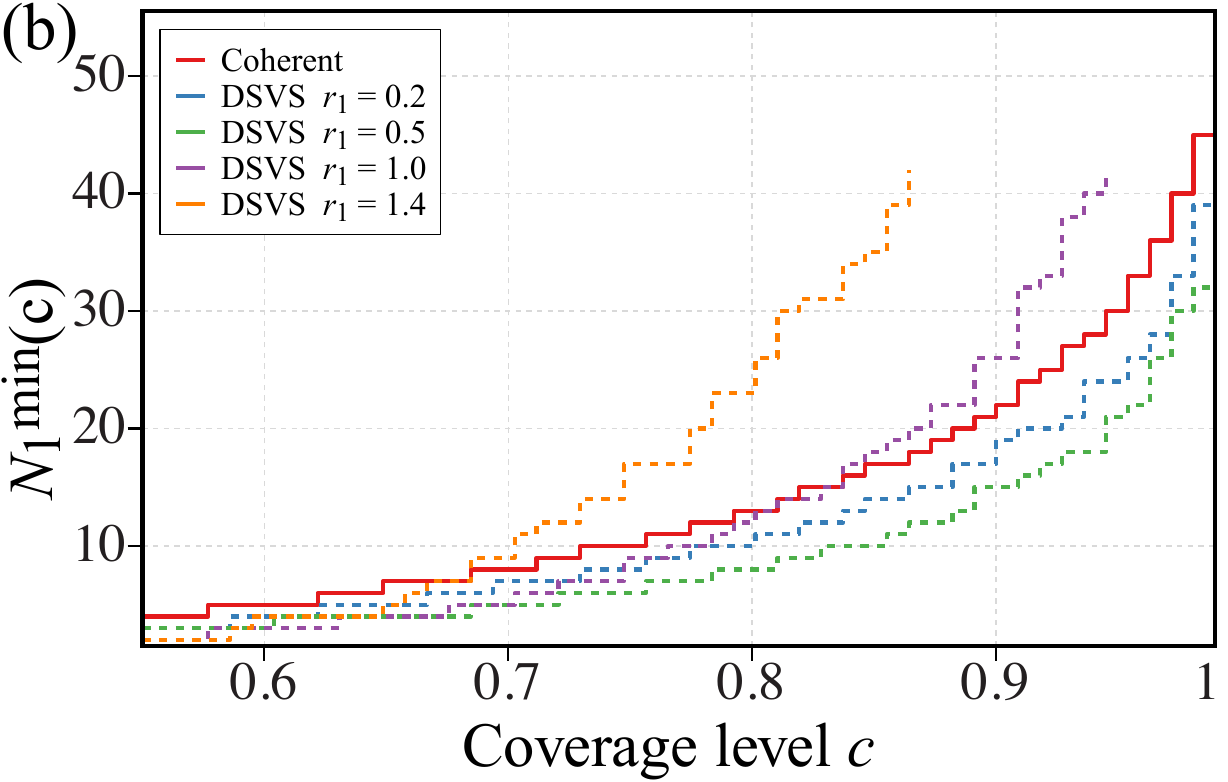}
    \vspace{-0.5em}

  \end{minipage}
  \caption{\label{fig:dsvs_stage1_thresholds} Thresholds in Stage~I for
    displaced squeezed probes with heterodyne measurement under phase matching
    $\chi_1=\pi$. Panel~(a) shows the one-probe energy threshold
    $E_{1,\mathrm{DSVS}}^{\min,(1)}(c;r_1)$. The threshold for coherent states
    $E_{1,\mathrm{coh}}^{\min,(1)}(c)$ is shown for comparison. Note that within
    this range, displaced squeezed states require higher energies to reach the
    target coverage. Panel~(b) shows the Stage~I sample complexity
    $N_{1,\mathrm{DSVS}}^{\min}(c;r_1,|\alpha_1|)$ at fixed $|\alpha_1|=0.35$,
    compared with the coherent-state probe count
    $N_{1,\mathrm{coh}}^{\min}(c;|\alpha_1|)$. Moderate squeezing can reduce the
    number of probes, although the dependence on $r_1$ is not monotone.}
\end{figure}
Moreover, a reduction in the number of probes is not, by itself, an energy
advantage. At fixed $|\alpha_1|$, the total energy in Stage~I needed to reach
coverage $c$ is
\begin{equation}
  E_{1,\mathrm{tot}}^{\min}(c;r_1,|\alpha_1|)
  =
  N_{1,\mathrm{DSVS}}^{\min}(c;r_1,|\alpha_1|)
  \bigl(|\alpha_1|^2+\sinh^2 r_1\bigr).
  \label{eq:dsvs_total_stage1_energy}
\end{equation}
while coherent-state counterpart is
\begin{equation}
  E_{1,\mathrm{tot,coh}}^{\min}(c;|\alpha_1|)
  =
  N_{1,\mathrm{coh}}^{\min}(c;|\alpha_1|)\,|\alpha_1|^2 .
  \label{eq:coh_total_stage1_energy}
\end{equation}

Figure~\ref{fig:dsvs_total_stage1_energy} shows the total energy required in
Stage~I for the same comparison as in
Figure~\ref{fig:dsvs_stage1_thresholds}(b). Although displaced squeezed probes
may require fewer observations, the energy saved by using fewer probes is not
enough to compensate for the required energy for generating squeezing in each
probe. We note that over the range of parameters studied here, the
coherent-state benchmark remains the least costly option in terms of total
energy in Stage~I to realize the coarse localization.

\begin{figure}[t]
  \centering
  \includegraphics[width=\columnwidth]{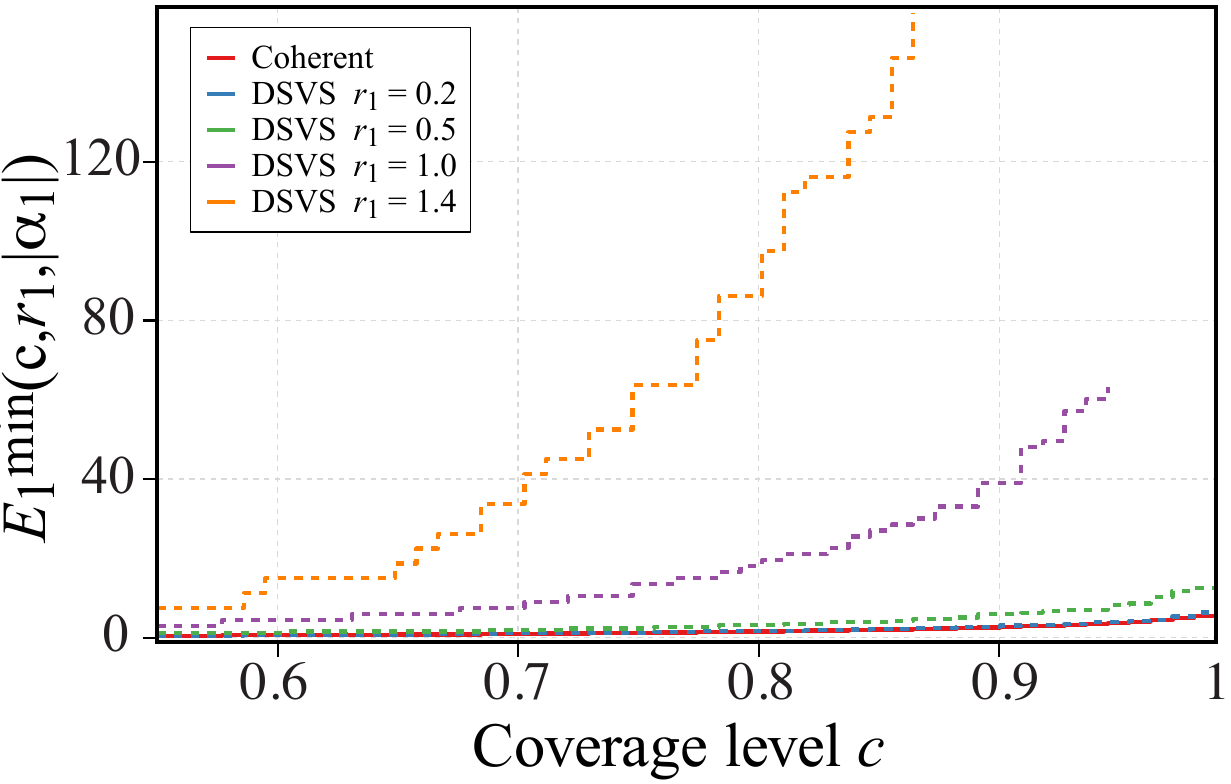}
  \caption{\label{fig:dsvs_total_stage1_energy} Total energy Stage~I induced by
    the sample-complexity threshold at fixed $|\alpha_1|=0.35$. The curves for
    displaced squeezed states show
    $E_{1,\mathrm{tot}}^{\min}(c;r_1,|\alpha_1|)$, compared to the case of
    coherent states $E_{1,\mathrm{tot,coh}}^{\min}(c;|\alpha_1|)$. The reduction
    in the number of probes in Stage~I does not compensate for the energy
    required in the per-probe squeezing $\sinh^2 r_1$. Thus, under heterodyne
    measurements, coherent-states probes remain less costly in total energy of
    Stage~I for coarse localization over the parameter range considered.}
\end{figure}

It is important to contrast these finite-sample conclusions with the local
asymptotic comparison. For $\chi_1=\pi$ and fixed displacement amplitude
$|\alpha_1|$,
\begin{equation}
  \mathcal I_{\mathrm{het}}^{\mathrm{DSVS}}
  >
  \mathcal I_{\mathrm{het}}^{\mathrm{coh}}
  \qquad
  \text{for every } r_1>0.
\end{equation}
Under the standard regularity assumptions for maximum-likelihood estimation,
this implies a smaller local asymptotic variance per probe at fixed
$|\alpha_1|$. However, this comparison, does not fix the total probe energy as
the displaced squeezed probe also carries the energy from squeezing $\sinh^2 r_1$.
The numerical results above show, however, that this local Fisher-information advantage does
not determine the finite-energy localization cost. Stage~I requires coverage of
a window of width $\pi/2$, not only estimation in an infinitesimal neighborhood
of the true phase. The finite-sample thresholds therefore depend on the
coverage event, the tail behavior of the Stage~I error distribution, and the
energy stored in squeezing. Within the heterodyne Stage~I scheme analyzed here,
displaced squeezing can reduce the number of probes needed for coarse
localization, but it does not reduce the total energy in Stage~I required over the
parameter range considered.

This distinction is relevant for the complete two-Stage protocol. A smaller
sample complexity  in Stage~I does not imply a smaller final error, because the
total energy budget is shared between coarse localization in Stage~I and local
phase estimation in Stage~II: $E_{\mathrm{tot}}=E_1+E_2.$
The following section studies the family of optimized two-Stage strategies, where
the two stages are jointly optimized under the total energy constraint.

\section{Two-stage design under a fixed total energy}
\label{sec:two_stage_fixed_budget}

In the preceding Sections \ref{sec:coherent_init} and \ref{sec:dsvs_init}, the
coarse, global localization stage (Stage~I) was analyzed in isolation for
different quantum probes. Those sections identified regimes in which displaced
squeezed probes improve the Stage~I coverage properties, but also showed that
this improvement does not necessarily translate into a lower energy cost for
Stage~I alone. For the full optimization of the two-Stage strategies, however,
we need to consider the total energy constraint $ E_1+E_2\le E$, including the
energy $E_1$ used for coarse localization in Stage~I and the energy $E_2$ for
local efficient estimation in Stage~II.

In these optimized strategies, Stage~I and Stage~II are jointly optimized to
minimize the final error bound, which has two contributions. One from the local
measurement in Stage~II, conditioned on successful coarse localization in
Stage~I. Another one from the overshoot penalty associated with incorrect choice
of localization window Stage~I within the circle. For a fixed total energy $E$,
such that $0<E_1<E$ and a prescribed number $N_2$ of squeezed-vacuum probes in
Stage~II, the lower error bound with coherent states in Stage~I is given by
\begin{equation}
  \mathcal B_{\mathrm{coh}}(E_1;E,N_2)
  \coloneq
  \frac{C_{\mathrm{coh}}(E_1)}
  {\mathcal H_2(E-E_1;N_2)}
  +
  G_{\mathrm{coh}}(E_1),
%  \qquad 0<E_1<E.
  \label{eq:coherent_two_stage_bound_fixed_N2}
\end{equation}
Here $C_{\mathrm{coh}}(E_1)$ and $G_{\mathrm{coh}}(E_1)$ are computed from the
Rician phase distribution obtained in Proposition~\ref{prop:rician_phase_law}.
For the protocol with DSVS probes in Stage~I, the corresponding objective
function is
\begin{multline}
  \mathcal{B}_{\mathrm{DSVS}}(N_1,|\alpha_1|,r_1;E,N_2)
  \coloneq
  \frac{
    C_{\mathrm{DSVS}}(N_1,|\alpha_1|,r_1)
  }{
    \mathcal{H}_2(E-E_1;N_2)
  }
  \\
  +
  G_{\mathrm{DSVS}}(N_1,|\alpha_1|,r_1).
  \label{eq:dsvs_two_stage_bound_fixed_N2}
\end{multline}
where $E_1=N_1(|\alpha_1|^2+\sinh^2 r_1)$ and $E_1<E$. The quantities
$C_{\mathrm{DSVS}}$ and $G_{\mathrm{DSVS}}$ are the coverage and overshoot terms
for the MLE in Stage~I . For $N_1=1$, these quantities are computed from the
integral representation in
Proposition~\ref{prop:dsvs_error_integral_representation}. For $N_1>1$, they are
evaluated from the full $N_1$-sample likelihood by Monte Carlo simulations.

We model Stage~II as an adaptive homodyne protocol using $N_2$ squeezed-vacuum
probes with squeezing strength $r_2$. If the energy allocated to Stage~II is
$E_2=E-E_1$, then $E_2=N_2\sinh^2 r_2,$ with
$r_2=\operatorname{arsinh}\sqrt{\frac{E_2}{N_2}}.$ The total QFI of the Stage~II
product family is therefore
\begin{equation}
\begin{split}
  \mathcal F_{Q,2}^{\mathrm{tot}}(E_2;N_2)
  &=
  8N_2\sinh^2 r_2\left(1+\sinh^2 r_2\right)  \\
  &=
  8\left(\frac{E_2^2}{N_2}+E_2\right)
  =
  8E_2\left(1+\frac{E_2}{N_2}\right).
\end{split}
\label{eq:stage2_qfi_fixed_N2}
\end{equation}
Thus, in Eqs.~\eqref{eq:coherent_two_stage_bound_fixed_N2} and
\eqref{eq:dsvs_two_stage_bound_fixed_N2}, the Stage~II QFI is evaluated at
$E_2=E-E_1$ for the prescribed value of $N_2$. This quantity should be
distinguished from the local benchmark obtained by optimizing over $N_2$.

For fixed $E_2$, the expression in Eq.~\eqref{eq:stage2_qfi_fixed_N2} is
maximized, over positive integers $N_2$, at $N_2=1$. This formal maximizer
corresponds to concentrating all the Stage~II energy into a single
squeezed-vacuum probe. It is useful as a local reference, but it does not
represent the multi-probe adaptive homodyne implementation considered here, and
it may require squeezing beyond what is experimentally accessible. Under the
fixed-energy constraint, increasing $N_2$ distributes the same energy among more
probes and therefore lowers the squeezing per probe. Thus, larger $N_2$ does
not improve the Fisher-information part of the lower bound by itself. Instead,
$N_2$ controls the tradeoff between having more adaptive refinement steps and
having more squeezing in each Stage~II probe.

The same expression also identifies the relevant scaling regimes. Suppose that
$E_2(E)=\Theta(E)$. If $N_2$ is fixed, then $\mathcal
F_{Q,2}^{\mathrm{tot}}(E_2(E);N_2)=\Theta(E^2),$ so the Stage~II contribution
can display Heisenberg scaling asymptotically. However, the onset of the
quadratic regime requires $E_2/N_2\gg 1$. When $E_2/N_2\ll 1$, the linear term
in Eq.~\eqref{eq:stage2_qfi_fixed_N2} dominates, and the Stage~II contribution
is shot-noise-like over the available energy range. More generally, if
$N_2(E)=\Theta(E^\gamma)$ with $0<\gamma<1$, then $\mathcal
F_{Q,2}^{\mathrm{tot}}(E_2(E);N_2(E)) = \Theta(E^{2-\gamma}),$ which gives an
intermediate scaling between the shot-noise and Heisenberg regimes.

\subsection{Numerical optimization}

We realize numerical studies of optimized two-Stage protocols for different
values of $N_2$. Specifically, we consider $N_2=1$, $N_2=5$, and $N_2=100$. The
first case is the formal local reference, the second is a finite-$N_2$
refinement that enters the quadratic regime at moderate energies, and the third
case illustrates a large fixed-$N_2$ available for implementing an efficient
(asymptotically optimal) phase measurement in Stage~II that remains
pre-asymptotic over the range of energies investigated.

For the two-stage protocol with coherent states in Stage~I, the optimization
reduces to a one-dimensional problem. For each pair $(E,N_2)$, we minimize
Eq.~\eqref{eq:coherent_two_stage_bound_fixed_N2} over $0<E_1<E$. In contrast,
for the two-stage protocol with DSVS as probes in Stage~I, one could minimize
Eq.~\eqref{eq:dsvs_two_stage_bound_fixed_N2} directly over
$(N_1,|\alpha_1|,r_1)$ under the constraint $N_1\left(|\alpha_1|^2+\sinh^2
  r_1\right)<E.$ For the latter case, numerical optimization becomes more
efficient by reparameterization of the energy constraint. To this end, we fix
$E$ and $N_1$. Then, the per-probe energy in Stage~I must satisfy
$|\alpha_1|^2+\sinh^2 r_1<\frac{E}{N_1}.$ Then, after choosing a value of $r_1$,
the largest admissible displacement amplitude becomes $\left(E/N_1-\sinh^2
  r_1\right)^{1/2}$. We therefore write $|\alpha_1| =
t\left(\frac{E}{N_1}-\sinh^2 r_1\right)^{1/2},$ with $0<t\le 1,$ and $0\le
r_1\le \operatorname{arsinh}\sqrt{\frac{E}{N_1}}.$ The variable $t$ specifies
the fraction of the largest admissible displacement amplitude used for the
chosen value of $r_1$. This change of variables turns the constrained search
over $(|\alpha_1|,r_1)$ into a bounded search over $(r_1,t)$, while still
representing any displaced squeezed probe in Stage~I with nonzero displacement
and energy below the total energy.

For each fixed $N_2$, we enumerate $N_1$ over the finite set $\{1,\ldots,40\}$.
The continuous minimization over $(r_1,t)$ is then performed using generalized
simulated annealing, followed by a bounded local refinement
\cite{TsallisStariolo1996,XiangEtAl2013}. The cutoff $N_1\le 40$ is a numerical
truncation of the discrete search space. In the reported runs, the selected
values of $N_1$ do not occur at the endpoint of this set, so the truncation is
not active in the explored range.

Figure~\ref{fig:two_stage_scaling_comparison} summarizes the results of the
optimized bounds for the two-stage strategies normalized to the ideal local
reference $\mathcal B_{\mathrm{loc}}(E)=\frac{1}{8E(E+1)}$. This reference
corresponds to assigning the total energy $E$ to a single squeezed-vacuum state
for a local estimation problem. We note that this bound does not correspond to a
full-period phase estimation problem, because it does not include localization
of Stage~I. The normalized ratio in
Figure~\ref{fig:two_stage_scaling_comparison}(a) therefore measures the
multiplicative cost of implementing a two-Stage full-period procedure relative
to this ideal local limit. Values close to one indicate that the optimized
protocol is close to the local squeezed vacuum benchmark, whereas larger values
reflect the combined cost of global localization, finite-sample effects, and
distributing the energy of Stage~II over several probes.

\begin{figure}[t]
    \centering
  \begin{minipage}{0.49\textwidth}
    \centering
    \includegraphics[width=\textwidth]{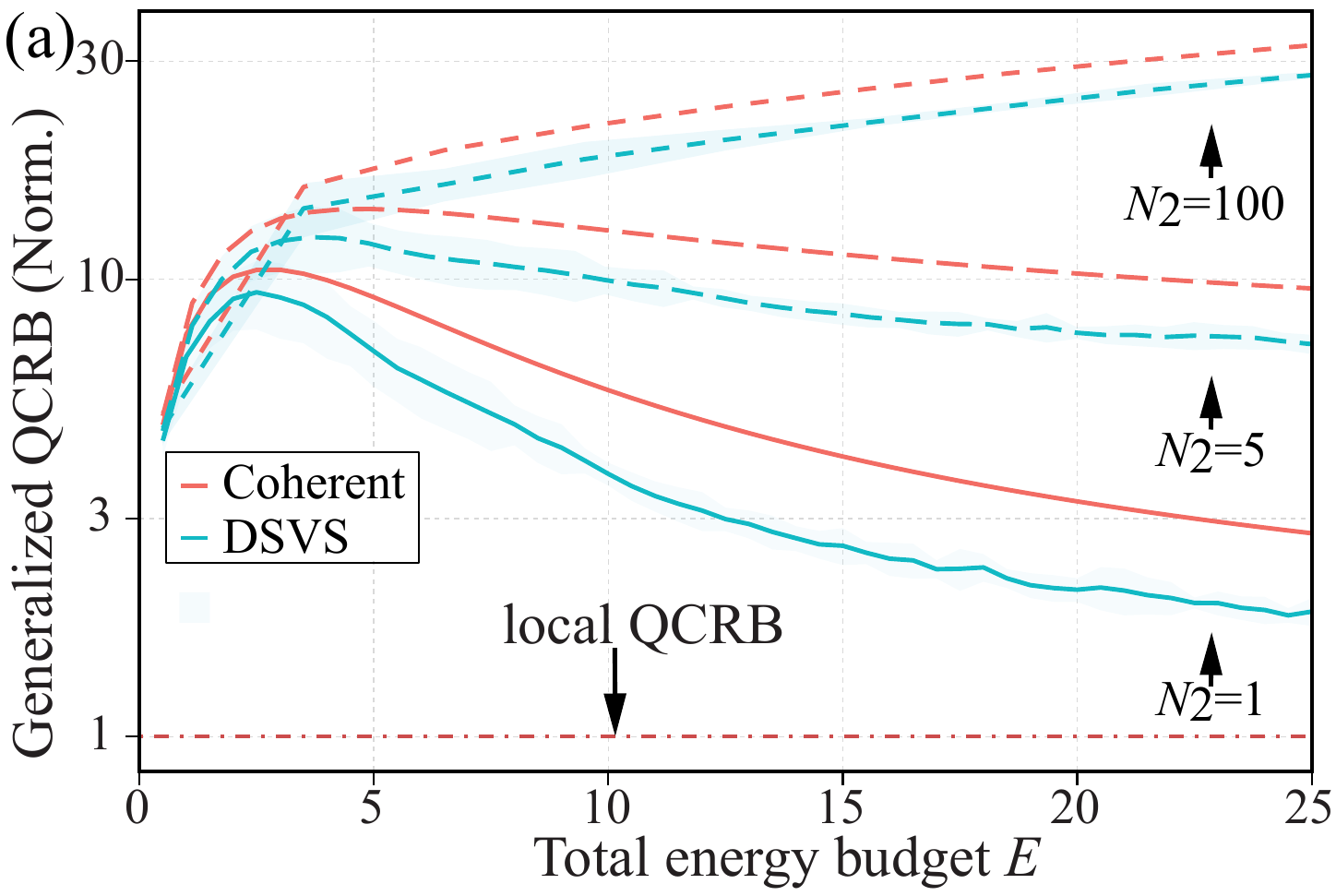}
    \vspace{-0.5em}

  \end{minipage}
  \hfill
  \begin{minipage}{0.49\textwidth}
    \centering
    \includegraphics[width=\textwidth]{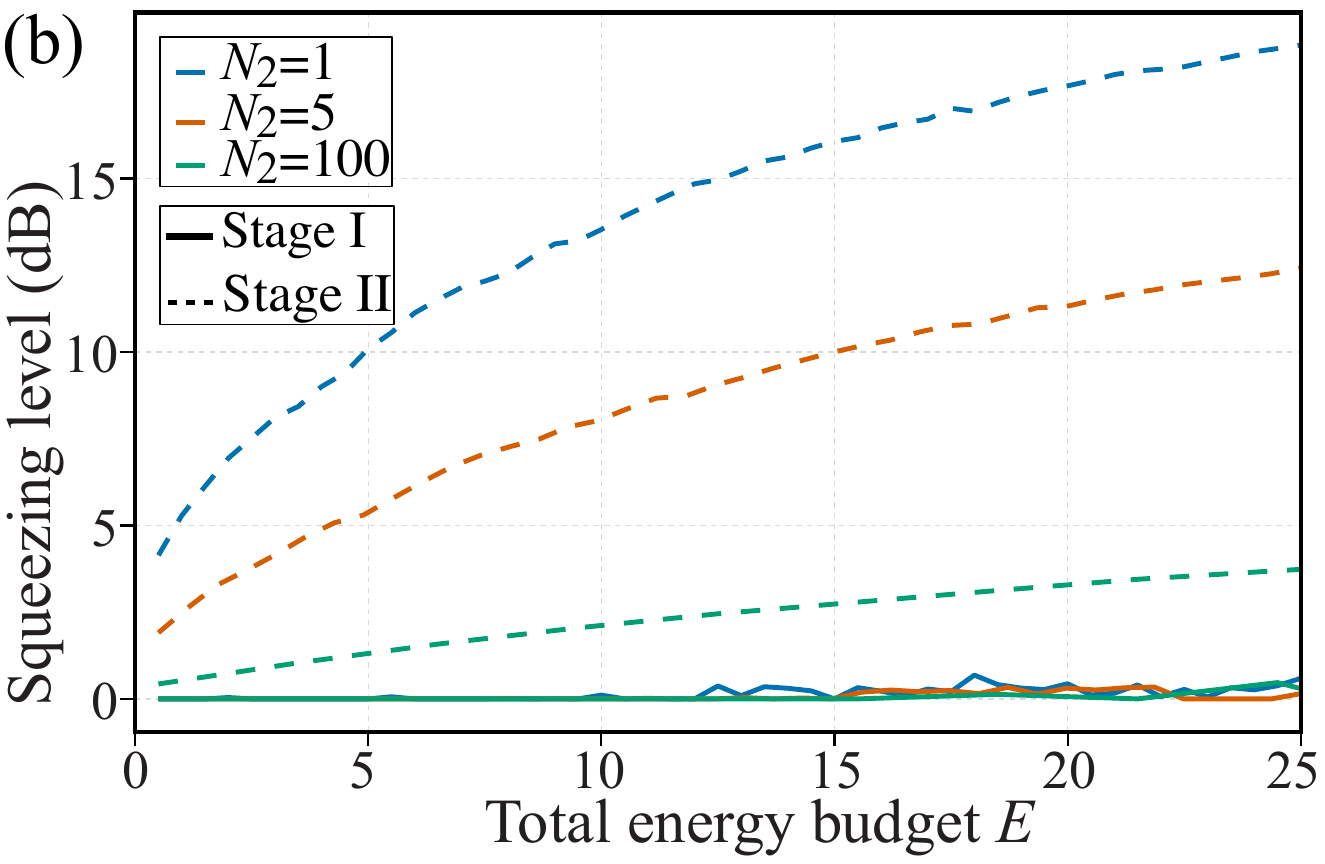}
    \vspace{-0.5em}

  \end{minipage}
  \caption{\label{fig:two_stage_scaling_comparison} Performance of two-Stage
    optimized strategy for full-period phase estimation under the total energy
    constraint. \textbf{(a)} Optimized generalized bound normalized to the ideal
    local reference $\mathcal B_{\mathrm{loc}}(E)=1/[8E(E+1)]$, which is
    obtained by assigning the total energy to a single squeezed-vacuum state for
    a local estimation problem. This quantum benchmark does not correspond to a
    full-period phase estimation problem. Different curves compare coherent and
    displaced squeezed probes in Stage~I for $N_2=1$, $N_2=5$, and $N_2=100$.
    Values above one give the multiplicative gap from the ideal local reference.
    \textbf{(b)} Required squeezing (in dBs) in Stage~I and Stage~II for the
    optimized strategies in \textbf{(a)}. The plotted values correspond to
    $20r/\log 10$, where $r$ is the squeezing parameter. Note that probe
    squeezing in Stage~I is smaller than in Stage~II, while the squeezing per
    probe in Stage~II decreases as $N_2$ increases.}
\end{figure}

Panels~(a) and (b) in Figure~\ref{fig:two_stage_scaling_comparison} are best
interpreted together. Panel~(a) compares the optimized bounds, while panel~(b)
shows the squeezing levels required by the corresponding optimized designs.
Thus, an experimental squeezing ceiling restricts which curves in panel~(a) are
physically accessible. For example, a maximum available squeezing of
$s_{\max}=10\,\mathrm{dB}$ corresponds to $r_{\max}=(\log 10/20)s_{\max}\simeq
1.15$, or to a per-probe squeezing energy $\sinh^2 r_{\max}\simeq 2.03$. Any
optimized design whose squeezing levels in Stage~I or Stage~II exceeds this
value would either be inaccessible under that constraint or would require
reoptimization with $r_1,r_2\le r_{\max}$.

This comparison also explains the role of $N_2$. Taking $N_2=1$ concentrates the
energy of Stage~II in a single squeezed-vacuum probe and gives the closest
approach to the local reference in panel~(a), but it may require a large
squeezing level per probe. Increasing $N_2$ distributes the energy of Stage~II
over more probes and lowers the squeezing required per probe, which can make the
protocol more compatible with a fixed squeezing ceiling. The cost is that, when
$N_2$ is too large relative to the available energy, the quadratic term in
$\mathcal H_2(E_2;N_2)$ is suppressed and the normalized bound moves farther
from the ideal local reference $\mathcal B_{\mathrm{loc}}(E)$. Under a squeezing
constraint, the practical choice is therefore the smallest value of $N_2$ for
which the required squeezing of Stage~II remains below the available
experimental limit over the desired energy range.

The optimized protocols with displaced squeezed probes in the Stage~I outperform
the protocols with coherent states over the investigated energy range and for
the three values of $N_2$ considered. We note, however, that this improvement is
not due to large squeezing in Stage~I. Indeed, panel~(b) shows that the
optimized Stage~I squeezing remains small compared with the Stage~II squeezing.
In the optimized designs, squeezing in Stage~I is mainly used to improve
localization on the circle while preserving full-period identifiability, whereas
most of the squeezing resource is allocated to the local estimation problem in
Stage~II.

Finally, we observe that protocols using displaced squeezed states in Stage~I
for global localization outperform protocols with coherent states for a total
energy up to 25, and for any $N_{2}$. For $N_2=1$, both protocols are closest to
the local quantum reference, because the energy available for Stage~II is
concentrated in a single squeezed-vacuum probe. The case $N_2=5$ retains a
finite sequence of adaptive homodyne measurements in Stage~II, while still
showing the effect of the quadratic term in Eq.~\eqref{eq:stage2_qfi_fixed_N2}
for total energies from 0-25. By contrast, for $N_2=100$, the per-probe energy
in the Stage~II for efficient phase estimation remains small throughout the
energy range considered, and the normalized curves increase relative to the
local benchmark for local parameter estimation. This decrease in performance
does not indicate a failure of the optimization. It rather reflects that the
measurement in Stage~II is still in the pre-asymptotic regime in which $\mathcal
H_2(E_2;N_2)$ is dominated by its linear term.

From a practical point of view, we can conclude that the best local scaling is
obtained by concentrating energy into few highly squeezed probes, while
experimental squeezing limits may favor a larger number of less squeezed probes
despite the larger gap from the local benchmark. The corresponding optimized
values of $N_1$, $r_1$, $r_2$, and the energy fractions assigned to each stage
are reported in Appendix~\ref{app:protocol_numerical_optimization}. These values
exemplify how the optimizer implements the tradeoff between full-period
localization and local Stage~II sensitivity.

Figure~\ref{fig:appendix_decomposition_all_N2} shows the two contributions to
the the optimized bound: the local refinement term and the overshoot term. The
comparison shows that the relative contributions of these terms strongly depend
$N_2$. For $N_2=1$, both the coherent and displaced squeezed protocols benefit
from the large QFI of the Stage~II, and the local term decreases rapidly with
the total energy. In this regime, the remaining gap is increasingly controlled
by the overshoot term of the Stage~I. For $N_2=5$, the same qualitative behavior
persists, although the local term is larger because the refinement energy is
distributed over several probes. For $N_2=100$, the local term remains
comparatively large over the plotted energy range, reflecting the fact that
$E_2/N_2$ is small and the Stage~II refinement has not yet reached the quadratic
regime.

We note that total bound for displaced squeezed states in Stage~I is
consistently smaller than in the case with coherent states, in agreement with
Figure~\ref{fig:two_stage_scaling_comparison}. The improvement comes from a
combined reduction of the local and overshoot contributions. In particular, the
overshoot term becomes very small for intermediate and large energies, showing
that the displaced squeezed localization stage suppresses large Stage~I errors
more effectively.

\begin{figure*}[t]
  \centering
  \includegraphics[width=.95\textwidth]{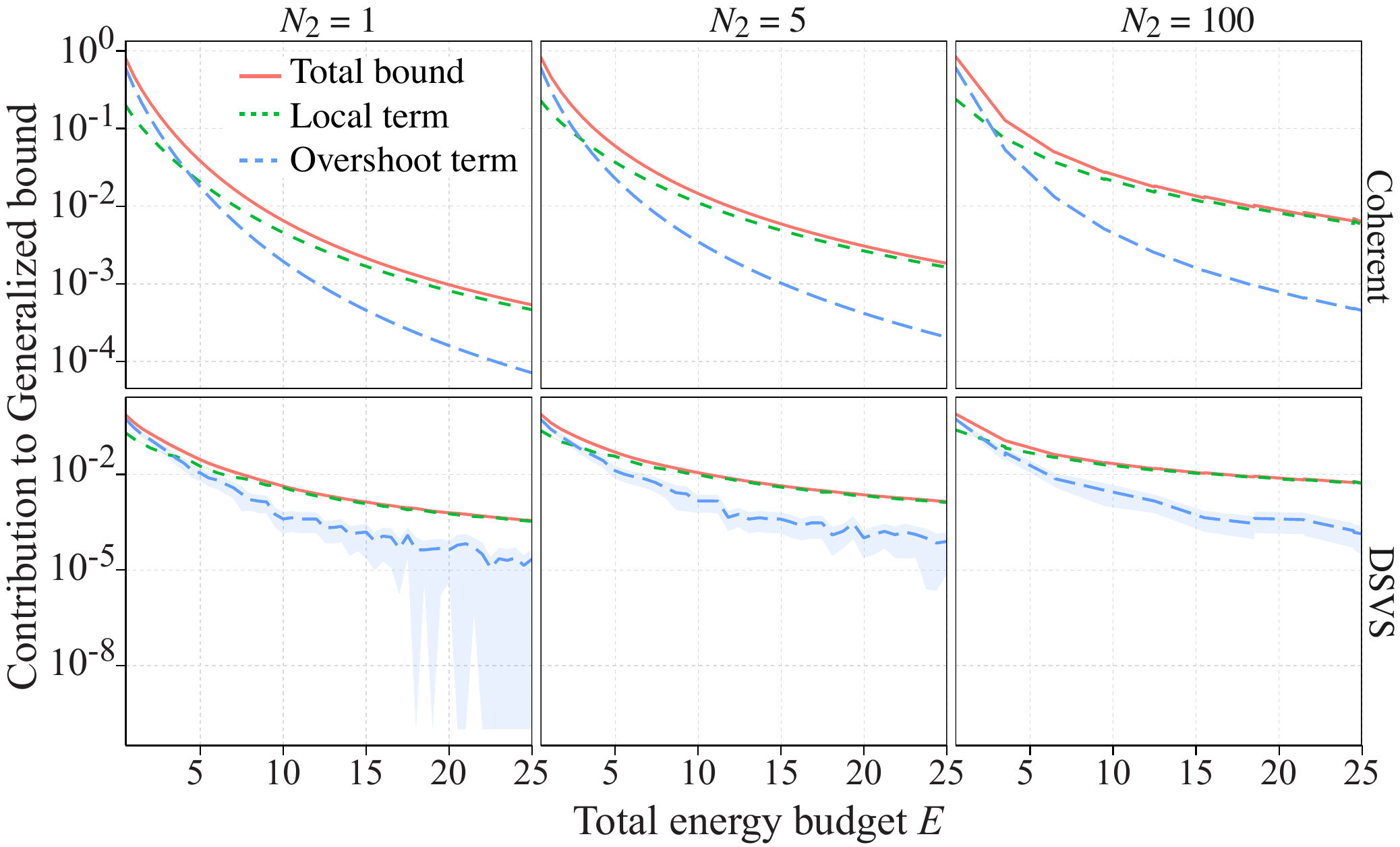}
  \caption{\label{fig:appendix_decomposition_all_N2} Decomposition of the
    optimized generalized bound for the coherent and displaced squeezed
    two-stage protocols. The columns correspond to $N_2=1$, $N_2=5$, and
    $N_2=100$, while the rows correspond to the coherent and displaced squeezed
    Stage~I probes. Each panel shows the total optimized bound, the local
    refinement term $C_1/\mathcal H_2$, and the overshoot term $G_1$ as
    functions of the total energy $E$. All quantities are shown on a logarithmic
    scale. The shaded regions in the displaced squeezed panels indicate Monte
    Carlo uncertainty in the numerical evaluation of the Stage~I coverage and
    overshoot terms. These fluctuations are visible mainly in the overshoot
    term, where rare localization failures dominate the estimate.}
\end{figure*}

\section{Summary and outlook}
\label{sec:summary_outlook}

We propose and analyzed two-stage Gaussian strategies for full-period optical
phase estimation under a fixed energy budget. Stage~I uses displaced squeezed
probes and heterodyne measurement to localize the phase on $[0,2\pi)$ to an
interval of length $\pi/2$. Stage~II uses squeezed-vacuum probes and adaptive
homodyne measurements to perform locally efficient estimation inside the
selected interval. The protocol keeps the family of probe states fixed within
each Stage, while separating the global task of phase localization from the
local task of efficient estimation.

The main theoretical result is a generalized Cramér-Rao bound for two-Stage
estimators constrained to the interval selected in Stage~I. The bound separates
the local Stage~II contribution from an overshoot penalty determined by Stage~I
localization failures. This term has no counterpart in purely local Cramér-Rao
theory and accounts for finite-energy events in which the selected interval does
not contain the true phase, producing an error of order one on the circle.

We computed the Stage~I coverage and overshoot terms for coherent probes and for
displaced squeezed probes under heterodyne measurements. In the case of coherent
states, the localization error follows a Rician phase distribution. In the
displaced squeezed case, the anisotropic heterodyne likelihood produces a
radius-dependent angular correction and an integral formula for the one-probe
error density. These formulas allow the two choices of probes in Stage~I to be
compared within the same two-Stage bound and under the same total energy
constraint.

Our numerical studies show that modest squeezing in Stage~I can improve the
final two-stage bound relative to coherent-state probes over the energy range
considered. In particular, we observed that Stage~I squeezing improves
correct-interval selection and suppresses the overshoot contribution, while most
of the squeezing resource remains allocated to the local problem in Stage~II.
Therefore, the relevant design criterion is the joint allocation of energy,
samples, and squeezing across both stages.

Several directions remain open. On the theoretical side, it would be useful to
compare the generalized Cramér-Rao bound with global bounds, such as
Bayesian-risk bounds \cite{Gill1995, Gill2013} or Ziv-Zakai-type inequalities
\cite{Ziv1969,Tsang2012}, and to derive non-asymptotic concentration bounds for
displaced squeezed heterodyne localization. A natural extension of this work is
to apply the same decomposition to periodic frequency-estimation protocols,
where longer interrogation times or entangled probes can improve local scaling
while reducing the unambiguous parameter range and introducing phase-slip errors
\cite{ZhengDoldeKolkowitz2023,PhysRevResearch.6.043230,Shaw2024}. In such
settings, an overshoot-type term could separate the local Fisher-information
contribution from the global cost of selecting the wrong branch of the periodic
parameter.

On the practical side, the framework presented here should be extended to
include optical loss, detector inefficiency, imperfect squeezing, and phase
drift, since these effects may change the optimal balance between coarse
localization and local estimation. In the finite-energy regime studied here,
optimized two-Stage Gaussian protocols remain within a factor of $3$-$30$ of the
idealized local squeezed-vacuum QCRB for $E\le 25$ photons and squeezing levels
up to $12\,\mathrm{dB}$, while displaced squeezing in Stage~I improves the bound
relative to coherent localization. This quantifies the main tradeoff: the local
advantage of squeezed probes becomes useful for full-period estimation only when
window coverage, overshoot suppression, and local sensitivity are optimized
together.

\section{Acknowledgments}
This work was supported by funding from the National Research Council of Canada,
the National Science Foundation (NSF) Grant PHY-2116246, the Department of
Energy (DOE) and Oak Ridge National Laboratory (ORNL) Subcontract CW42943.
L.M.D. acknowledge partial support from project DGAPA UNAM IN 117925.

%\newpage

\appendix
%\appendixsectionformat
\section{Proof for the Generalized two-stage Cramér-Rao bound}
\label{app:proof_main_theorem}

In this appendix we give the proof of Theorem~\ref{thm:generalized_crb}.

\begin{proof}
  The proof splits the risk according to whether the Stage~I window contains
  $\theta_0$. On the coverage event, the problem is expressed in the lifted
  coordinate interval and the conditional Cramér-Rao inequality bounds the
  Stage~II contribution. On the complementary event, the estimator is
  constrained to a window that excludes the true phase, so
  Lemma~\ref{lem:overshoot_lb} bounds the error by the overshoot distance. These
  two bounds give the local Stage~II term and the Stage~I overshoot penalty.

We first split the mean-square error according to the Stage~I coverage set:
\begin{equation}
\begin{split}
  \mathbb E_{\theta_0}\!\left[
    d\left(\hat\theta,\theta_0\right)^2
  \right]
  &=
  \mathbb E_{\theta_0}\!\left[
    d\left(\hat\theta,\theta_0\right)^2
    \mathbf 1_{\mathcal C_{\theta_0}}(Z)
  \right]  \\
  &\quad+
  \mathbb E_{\theta_0}\!\left[
    d\left(\hat\theta,\theta_0\right)^2
    \mathbf 1_{\mathcal C_{\theta_0}^c}(Z)
  \right].
\end{split}
\label{eq:proof_risk_split}
\end{equation}

We first bound the contribution from the complement of the coverage event. By
Lemma~\ref{lem:overshoot_lb}, for every realization of $(Z,X)$,
\begin{equation}
  d\left(\hat\theta(Z,X),\theta_0\right)^2
  \ge
  \Delta_{\theta_0}(Z)^2.
\end{equation}
Moreover, $\Delta_{\theta_0}(Z)=0$ whenever $Z\in\mathcal C_{\theta_0}$, since
in that case the selected window contains the true phase. Hence
\begin{equation}
  \Delta_{\theta_0}(Z)^2
  =
  \Delta_{\theta_0}(Z)^2
  \mathbf 1_{\mathcal C_{\theta_0}^c}(Z),
\end{equation}
and therefore
\begin{equation}
\begin{split}
  \mathbb E_{\theta_0}\!\left[
    d\left(\hat\theta,\theta_0\right)^2
    \mathbf 1_{\mathcal C_{\theta_0}^c}(Z)
  \right]
  &\ge
  \mathbb E_{\theta_0}\!\left[
    \Delta_{\theta_0}(Z)^2
    \mathbf 1_{\mathcal C_{\theta_0}^c}(Z)
  \right]  \\
  &=
  \mathbb E_{\theta_0}\!\left[
    \Delta_{\theta_0}(Z)^2
  \right].
\end{split}
\label{eq:proof_overshoot_contribution}
\end{equation}

It remains to bound the contribution from the coverage event. Fix a Stage~I
outcome $z$ such that $z\in\mathcal C_{\theta_0}$, equivalently $\theta_0\in
W(z)$. Since the final estimator is constrained to the selected window,
$\hat\theta(z,X)\in W(z)$, it has a unique lift $\widetilde{\hat\theta}_z(X)\in
I_z$. By the construction of the lifted coordinate interval,
$d\left(\hat\theta(z,X),\theta_0\right) = \left|
  \widetilde{\hat\theta}_z(X)-\tilde\theta_0(z) \right|.$ For this fixed $z$,
the adaptive Stage~II policy induces the conditional model
$p_{2,z}(x;\vartheta)$ on the interval $I_z$. By the local regularity
assumptions on the coverage event, Assumptions~\ref{ass:local_crb}, this model
is regular in a neighborhood of $\tilde\theta_0(z)$ and the lifted estimator is
locally unbiased at $\tilde\theta_0(z)$. The conditional Cramér-Rao inequality
therefore gives
\begin{widetext}
\begin{equation}
  \mathbb{E}_{\theta_0}\!\left[
    d\left(\hat\theta(z,X),\theta_0\right)^2
    \Bigm| Z=z
  \right]
  =
  \mathbb{E}_{\theta_0}\!\left[
    \left(
      \widetilde{\hat\theta}_z(X)-\tilde\theta_0(z)
    \right)^2
    \Bigm| Z=z
  \right] \ge
  \frac{1}{\mathcal{I}_2(\theta_0;z)}.
\label{eq:proof_conditional_crb}
\end{equation}
\end{widetext}
This inequality is used only for Stage~I outcomes in $\mathcal C_{\theta_0}$.
Equivalently,
\begin{equation}
  \mathbf 1_{\mathcal C_{\theta_0}}(z)
  \mathbb E_{\theta_0}\!\left[
    d\left(\hat\theta(z,X),\theta_0\right)^2
    \Bigm| Z=z
  \right]
  \ge
  \frac{
    \mathbf 1_{\mathcal C_{\theta_0}}(z)
  }{
    \mathcal I_2(\theta_0;z)
  }.
  \label{eq:proof_conditional_crb_indicator}
\end{equation}

Integrating Eq.~\eqref{eq:proof_conditional_crb_indicator} with respect to the
marginal distribution of $Z$ under $\theta_0$ and using the tower property of
conditional expectation yields
\begin{equation}
\begin{split}
  \mathbb E_{\theta_0}\!\left[
    d\left(\hat\theta,\theta_0\right)^2
    \mathbf 1_{\mathcal C_{\theta_0}}(Z)
  \right]
  &=
  \mathbb E_{\theta_0}\!\left[
    \mathbf 1_{\mathcal C_{\theta_0}}(Z)
    \mathbb E_{\theta_0}\!\left[
      d\left(\hat\theta,\theta_0\right)^2
      \Bigm| Z
    \right]
  \right]  \\
  &\ge
  \mathbb E_{\theta_0}\!\left[
    \frac{
      \mathbf 1_{\mathcal C_{\theta_0}}(Z)
    }{
      \mathcal I_2(\theta_0;Z)
    }
  \right].
\end{split}
\label{eq:proof_local_contribution}
\end{equation}

Combining the risk decomposition in Eq.~\eqref{eq:proof_risk_split} with the
bounds in Eqs.~\eqref{eq:proof_overshoot_contribution} and
\eqref{eq:proof_local_contribution} gives
\begin{equation}
  \mathrm{MSE}\left(\hat\theta;\theta_0\right)
  \ge
  \mathbb E_{\theta_0}\!\left[
    \frac{
      \mathbf 1_{\mathcal C_{\theta_0}}(Z)
    }{
      \mathcal I_2(\theta_0;Z)
    }
  \right]
  +
  \mathbb E_{\theta_0}\!\left[
    \Delta_{\theta_0}(Z)^2
  \right],
\end{equation}
which is Eq.~\eqref{eq:generalized_crb_main}.
\end{proof}

\section{Proof of the coherent Stage~I Rician phase distribution}
\label{app:proof_rician_phase_law}

In this appendix we give the full derivation of
Proposition~\ref{prop:rician_phase_law}.

\begin{proof}[Proof of Proposition~\ref{prop:rician_phase_law}]
  The proof has three steps. First, we reduce the $N_1$ heterodyne outcomes to
  the sufficient statistic $B_{N_1}$ for $\theta$ and rotate it so that its mean
  is real and positive. Second, we express the rotated complex Gaussian variable
  in polar coordinates, so that the signed wrapped error is its angular
  coordinate. Third, we integrate out the radial coordinate to obtain the
  marginal density of the angle.

  % By Eq.~\eqref{eq:equivalent_single_shot_coherent},

  Firs, we note that the sufficient statistic $B_{N_1} =
  \frac{1}{\sqrt{N_1}}\sum_{k=1}^{N_1}\beta_k$ for $\theta$ based on the
  heterodyne measurement outcomes is complex normal with distribution
\begin{equation}
  B_{N_1}
  \sim
  \mathcal{CN}\!\left(
    \sqrt{E_{\mathrm{coh}}}\,e^{i(-\theta+\phi_1)},1
  \right).
  \label{eq:proof_rician_B_distribution}
\end{equation}
Then, the MLE of Stage~I is $\hat\theta_1 = \Arg(B_{N_1})-\phi_1,$ understood
modulo $2\pi$. So, the signed wrapped Stage~I error $\varepsilon_1\in(-\pi,\pi]$
is
\begin{equation}
  \varepsilon_1
  =
  \Arg\!\left(e^{i(\hat\theta_1-\theta)}\right)
  =
  \Arg\!\left(e^{-i(-\theta+\phi_1)}B_{N_1}\right),
  \label{eq:proof_rician_signed_error}
\end{equation}
and the corresponding circular distance is $d\left(\hat\theta_1,\theta
\right)=|\varepsilon_1|.$

Now, we observe that the rotated statistic $B_{N_1}^{\theta,\phi_1} =
e^{-i(\theta+\phi_1)}B_{N_1}$ follows a complex normal distribution with
parameters $\left(\sqrt{E_{\mathrm{coh}}},1\right)$, that is,
\begin{equation}
  \widetilde B_{N_1}^{\theta,\phi_1}
  \sim
  \mathcal{CN}\!\left(\sqrt{E_{\mathrm{coh}}},1\right).
  \label{eq:proof_rician_rotated_statistic}
\end{equation}
%In particular, the distribution of $\varepsilon_1$ is independent of $\theta$
%and $\phi_1$.
Then, the signed wrapped error is
\begin{equation}
  \varepsilon_1
  =
  \Arg\!\left(e^{i(\hat\theta_1-\theta)}\right)
  =
  -\Arg\!\left(\widetilde B_{N_1}^{\theta,\phi_1}\right)
  \quad \text{mod }2\pi .
\end{equation}
Since $\widetilde B_{N_1}^{\theta,\phi_1}$ has a real positive mean and circular
Gaussian noise, its angular distribution is symmetric about zero. Therefore,
$\varepsilon_1$ has the same density as $\Arg(\widetilde
B_{N_1}^{\theta,\phi_1})$.

We now write $B_{N_1}^{\theta,\phi_1}=U+iV$. As the complex normal variance is
one, the real and imaginary parts are independent, with $U\sim\mathcal
N(\sqrt{E_{\mathrm{coh}}},1/2)$ and $V\sim\mathcal N(0,1/2)$. Therefore, the
joint density of $(U,V)$ with respect to Lebesgue measure on $\mathbb R^2$ is
\begin{equation}
  f_{U,V}(u,v)
  =
  \frac{1}{\pi}
  \exp\!\left(
    -\left[(u-\sqrt{E_{\mathrm{coh}}})^2+v^2\right]
  \right).
  \label{eq:proof_rician_joint_cartesian}
\end{equation}

Since $\mathbb P(B_{N_1}^{\theta,\phi_1}=0)=0$, we may use polar coordinates on
$\mathbb R^2\setminus\{(0,0)\}$. Write $u=r\cos\varepsilon,$
$v=r\sin\varepsilon,$ with $r\in(0,\infty),$ and $\varepsilon\in(-\pi,\pi].$ The
Jacobian determinant is $\left| \frac{\partial(u,v)}{\partial(r,\varepsilon)}
\right| = r.$ Hence, the joint density of $(R,\varepsilon_1)$ with respect to
$dr\,d\varepsilon$ is
\begin{align}
  f_{R,\varepsilon}(r,\varepsilon)
  &=
  f_{U,V}(r\cos\varepsilon,r\sin\varepsilon)\,r
  \notag\\
  &=
  \frac{r}{\pi}
  \exp\!\left(
    -\left[
      (r\cos\varepsilon-\sqrt{E_{\mathrm{coh}}})^2
      +
      r^2\sin^2\varepsilon
    \right]
  \right).
  \label{eq:proof_rician_joint_polar_1}
\end{align}
Expanding the exponent gives
\begin{equation}
  (r\cos\varepsilon-\sqrt{E_{\mathrm{coh}}})^2
  +
  r^2\sin^2\varepsilon
  =
  r^2
  -
  2r\sqrt{E_{\mathrm{coh}}}\cos\varepsilon
  +
  E_{\mathrm{coh}}.
\end{equation}
Therefore,
\begin{equation}
  f_{R,\varepsilon}(r,\varepsilon)
  =
  \frac{r}{\pi}
  \exp\!\left(
    -r^2
    +
    2r\sqrt{E_{\mathrm{coh}}}\cos\varepsilon
    -
    E_{\mathrm{coh}}
  \right).
  \label{eq:proof_rician_joint_polar_2}
\end{equation}

Integrating out $r$ gives the marginal density of $\varepsilon_1$ with respect
to Lebesgue measure on $(-\pi,\pi]$:
\begin{align}
  f_{\varepsilon}(\varepsilon;E_{\mathrm{coh}})
  &=
  \int_0^\infty
  f_{R,\varepsilon}(r,\varepsilon)\,dr
  \notag\\
  &=
  \frac{e^{-E_{\mathrm{coh}}}}{\pi}
  \int_0^\infty
  r\,
  e^{-r^2+2(\sqrt{E_{\mathrm{coh}}}\cos\varepsilon)r}\,dr.
  \label{eq:proof_rician_marginal_start}
\end{align}
Let $a=\sqrt{E_{\mathrm{coh}}}\cos\varepsilon.$ Then $-r^2+2ar = a^2-(r-a)^2,$
and therefore,
\begin{align}
  \int_0^\infty r\,e^{-r^2+2ar}\,dr
  &=
  e^{a^2}\int_0^\infty r\,e^{-(r-a)^2}\,dr
  \notag\\
  &=
  e^{a^2}\int_{-a}^{\infty}(t+a)e^{-t^2}\,dt
  \notag\\
  &=
  \frac12
  +
  \frac{a\sqrt{\pi}}{2}e^{a^2}\left(1+\erf(a)\right).
  \label{eq:proof_rician_radial_integral}
\end{align}
Substituting this expression into Eq.~\eqref{eq:proof_rician_marginal_start},
and using $-E_{\mathrm{coh}}+a^2=-E_{\mathrm{coh}}\sin^2\varepsilon$, we obtain
% \begin{equation}
%   f_{\varepsilon}(\varepsilon;E_{\mathrm{coh}})
%   =
%   \frac{e^{-E_{\mathrm{coh}}}}{2\pi}
%   +
%   \frac{\sqrt{E_{\mathrm{coh}}}\cos\varepsilon}{2\sqrt{\pi}}\,
%   e^{-E_{\mathrm{coh}}\sin^2\varepsilon}
%   \Bigl(
%     1+\erf(\sqrt{E_{\mathrm{coh}}}\cos\varepsilon)
%   \Bigr).
%   \label{eq:proof_rician_density_final}
% \end{equation}
\begin{multline}
  f_{\varepsilon}(\varepsilon;E_{\mathrm{coh}})
  =
  \frac{e^{-E_{\mathrm{coh}}}}{2\pi}
  +
  \frac{\sqrt{E_{\mathrm{coh}}}\cos\varepsilon}{2\sqrt{\pi}}\,
  e^{-E_{\mathrm{coh}}\sin^2\varepsilon}
  \\
  \times
  \left(
    1+\erf(\sqrt{E_{\mathrm{coh}}}\cos\varepsilon)
  \right).
  \label{eq:proof_rician_density_final}
\end{multline}
This is Eq.~\eqref{eq:rician_phase_density_coherent}.

% Finally, using $1+\erf(x)=2\Phi(\sqrt{2}\,x),$ we obtain the equivalent
% representation
% \begin{equation}
%   f_{\varepsilon}(\varepsilon;E_{\mathrm{coh}})
%   =
%   \frac{e^{-E_{\mathrm{coh}}}}{2\pi}
%   +
%   \frac{\sqrt{E_{\mathrm{coh}}}\cos\varepsilon}{\sqrt{\pi}}\,
%   e^{-E_{\mathrm{coh}}\sin^2\varepsilon}
%   \Phi\!\left(
%     \sqrt{2E_{\mathrm{coh}}}\cos\varepsilon
%   \right),
%   \label{eq:proof_rician_density_phi}
% \end{equation}
% which is Eq.~\eqref{eq:rician_phase_density_coherent_phi}. This completes the
% proof.
\end{proof}

\section{Details for the one-probe displaced squeezed error distribution}
\label{app:dsvs_error_density}

Throughout this appendix, we write $[x]_{2\pi}\coloneqq
\Arg(e^{ix})\in(-\pi,\pi]$ for the representative of an angle on the principal
branch.

\begin{proof}[Proof of Proposition~\ref{prop:dsvs_error_integral_representation}]

  The proof consists of conditioning on the observed radius $S_Y=s$. For each
  fixed radius, the maximum-likelihood correction shifts the angular coordinate
  $\Gamma$ by $\eta_\ast(s)$ on the circle. As circular translations preserve
  Lebesgue measure, the conditional density of the signed error is obtained by a
  shift of the conditional density of $\Gamma$. Integrating over the radius then
  gives the desired marginal density.

  By construction, $Y=S_Yv(\Gamma),$ $S_Y\ge 0,$ and $\Gamma\in(-\pi,\pi]$.
  Moreover, the joint density of $(S_Y,\Gamma)$ with respect to $ds\,d\gamma$ is
  $f_{S_Y,\Gamma}$ from Eq.~\eqref{eq:dsvs_joint_polar_density}. For each fixed
  $s\ge 0$, the MLE selects the angular correction $\eta_\ast(s)$ maximizing
  $\Lambda_s(\eta)$. Hence, the signed wrapped error is $\varepsilon_1 =
  [\Gamma-\eta_\ast(S_Y)]_{2\pi}.$

  Fix $s$ outside a null set for the conditional distribution of $\Gamma$ given
  $S_Y=s$. The map $\gamma\longmapsto [\gamma-\eta_\ast(s)]_{2\pi}$ is a
  translation on the circle and therefore preserves Lebesgue measure on
  $(-\pi,\pi]$. Consequently, the conditional density of $\varepsilon_1$ given
  $S_Y=s$ is
\begin{equation}
  f_{\varepsilon\mid S_Y}(\varepsilon\mid s)
  =
  f_{\Gamma\mid S_Y}
  \!\left(
    [\varepsilon+\eta_\ast(s)]_{2\pi}\mid s
  \right),
  \label{eq:app_dsvs_conditional_error_density}
\end{equation}
for $\varepsilon\in(-\pi,\pi]$.

Multiplying Eq.~\eqref{eq:app_dsvs_conditional_error_density} by the marginal
density of $S_Y$ and integrating over $s$ gives
\begin{align}
  f_{\varepsilon}^{\mathrm{DSVS}}(\varepsilon)
  &=
  \int_0^\infty
  f_{\varepsilon\mid S_Y}(\varepsilon\mid s) f_{S_Y}(s)\,ds
  \notag\\
  &=
  \int_0^\infty
  f_{\Gamma\mid S_Y}
  \!\left(
    [\varepsilon+\eta_\ast(s)]_{2\pi}\mid s
  \right)
  f_{S_Y}(s)\,ds
  \notag\\
  &=
  \int_0^\infty
  f_{S_Y,\Gamma}
  \!\left(
    s,
    [\varepsilon+\eta_\ast(s)]_{2\pi}
  \right)\,ds .
  \label{eq:app_dsvs_error_density}
\end{align}
This proves Eq.~\eqref{eq:dsvs_error_density_integral}.
\end{proof}

\vspace{-2em}

\section{Coherent limit of the displaced squeezed error density}
\label{app:dsvs_coherent_limit}

We show that Proposition~\ref{prop:dsvs_error_integral_representation} reduces
to the Rician phase distribution when $r_1=0$. First we note that by
substituting into Eq.~\eqref{eq:app_dsvs_error_density} the expression for
$f_{S_Y,\Gamma}$ from Eq.~\eqref{eq:dsvs_joint_polar_density} gives
\begin{widetext}
\begin{equation}
  f_{\varepsilon}^{\mathrm{DSVS}}(\varepsilon)
  =
  \int_0^\infty
  \frac{s}{2\pi\sqrt{\det\Sigma_{\delta_1}}}
  \exp\!\left[
    -\frac{1}{2}
    \left(
      s\,v(\varepsilon+\eta_\ast(s))-\mu_1 e_1
    \right)^\top
    \Sigma_{\delta_1}^{-1}
    \left(
      s\,v(\varepsilon+\eta_\ast(s))-\mu_1 e_1
    \right)
  \right] ds.
  \label{eq:dsvs_error_density_integral_explicit}
\end{equation}
\end{widetext}
% \begin{multline}
%   f_{\varepsilon}^{\mathrm{DSVS}}(\varepsilon)
%   =
%   \int_0^\infty
%   \frac{s}{2\pi\sqrt{\det\Sigma_{\delta_1}}}
%   \exp\!\left[
%     -\frac{1}{2}
%     \left(
%       s\,v(\varepsilon+\eta_\ast(s))-\mu_1 e_1
%     \right)^\top
%     \Sigma_{\delta_1}^{-1}
%   \right.\\
%   \left.
%     \times
%     \left(
%       s\,v(\varepsilon+\eta_\ast(s))-\mu_1 e_1
%     \right)
%   \right] ds.
%   \label{eq:dsvs_error_density_integral_explicit}
% \end{multline}
Since $v(\gamma)$ is $2\pi$-periodic, the explicit expression may be written
with $v(\varepsilon+\eta_\ast(s))$ without changing the value of the integrand.

Then, in the limit of $r \to 0$, Eq.~\eqref{eq:dsvs_heterodyne_eigenvalues}
gives $\lambda_{1,-}=\lambda_{1,+}=1$, and hence $\Sigma_{\delta_1}=\mathbb
I_2$. The restricted log-likelihood in Eq.~\eqref{eq:dsvs_radial_profile}
becomes
\begin{equation}
  \Lambda_s(\eta)
  =
  -\frac12
  \bigl\|s\,v(\eta)-\mu_1e_1\bigr\|^2 =
  -\frac12
  \left(
    s^2+\mu_1^2-2s\mu_1\cos\eta
  \right).
  \label{eq:app_coherent_limit_profile}
\end{equation}
For $s>0$ and $\mu_1>0$, this is maximized at $\eta=0$ modulo $2\pi$. At
$s=0$, the value of $\eta$ is irrelevant, so we choose the same branch
$\eta_\ast(0)=0$. Thus, $\eta_\ast(s)\equiv 0$.

Equation~\eqref{eq:dsvs_error_density_integral_explicit} then gives
\begin{equation}
  f_{\varepsilon}^{\mathrm{DSVS}}(\varepsilon)
  =
  \int_0^\infty
  \frac{s}{2\pi}
  \exp\!\left[
    -\frac12
    \bigl\|s\,v(\varepsilon)-\mu_1e_1\bigr\|^2
  \right]ds.
  \label{eq:app_coherent_limit_integral_1}
\end{equation}
As $\bigl\|s\,v(\varepsilon)-\mu_1e_1\bigr\|^2 =
s^2+\mu_1^2-2s\mu_1\cos\varepsilon,$ we obtain
\begin{equation}
  f_{\varepsilon}^{\mathrm{DSVS}}(\varepsilon)
  =
  \int_0^\infty
  \frac{s}{2\pi}
  \exp\!\left[
    -\frac12
    \left(
      s^2+\mu_1^2-2s\mu_1\cos\varepsilon
    \right)
  \right]ds.
  \label{eq:app_coherent_limit_integral_2}
\end{equation}
Set $u=s/\sqrt{2}$. As $\mu_1=\sqrt{2}|\alpha_1|$, this becomes
% \begin{equation}
%   f_{\varepsilon}^{\mathrm{DSVS}}(\varepsilon)
%   =
%   \frac{e^{-|\alpha_1|^2}}{\pi}
%   \int_0^\infty
%   u
%   \exp\!\left[
%     -u^2+2u|\alpha_1|\cos\varepsilon
%   \right]du.
%   \label{eq:app_coherent_limit_integral_3}
% \end{equation}
\begin{equation}
\begin{split}
  f_{\varepsilon}^{\mathrm{DSVS}}(\varepsilon)
  &=
   \frac{e^{-|\alpha_1|^2}}{\pi}
  \int_0^\infty
  u
  \exp\!\left[
    -u^2+2u|\alpha_1|\cos\varepsilon
  \right]du \\
  &=
  \frac{e^{-|\alpha_1|^2}}{2\pi}
  +
  \frac{|\alpha_1|\cos\varepsilon}{2\sqrt{\pi}}\,
  e^{-|\alpha_1|^2\sin^2\varepsilon}  \\
  &\quad\times
  \left(
    1+\erf\!\left(|\alpha_1|\cos\varepsilon\right)
  \right),
\end{split}
\label{eq:app_coherent_limit_density}
\end{equation}
which agrees with Eq.~\eqref{eq:rician_phase_density_coherent} for
$E_{\mathrm{coh}}=|\alpha_1|^2$.

% This is the radial integral obtained in the proof of
% Proposition~\ref{prop:rician_phase_law}, with
% $E_{\mathrm{coh}}=|\alpha_1|^2$ in the one-probe case. Therefore,
% \begin{equation}
% \begin{split}
%   f_{\varepsilon}^{\mathrm{DSVS}}(\varepsilon)
%   &=
%    \frac{e^{-|\alpha_1|^2}}{\pi}
%   \int_0^\infty
%   u
%   \exp\!\left[
%     -u^2+2u|\alpha_1|\cos\varepsilon
%   \right]du \\
%   &=
%   \frac{e^{-|\alpha_1|^2}}{2\pi}
%   +
%   \frac{|\alpha_1|\cos\varepsilon}{2\sqrt{\pi}}\,
%   e^{-|\alpha_1|^2\sin^2\varepsilon}  \\
%   &\quad\times
%   \left(
%     1+\erf\!\left(|\alpha_1|\cos\varepsilon\right)
%   \right).
% \end{split}
% \label{eq:app_coherent_limit_density}
% \end{equation}
% which agrees with Eq.~\eqref{eq:rician_phase_density_coherent} for
% $E_{\mathrm{coh}}=|\alpha_1|^2$.

\vspace{-1em}

\section{Details of the numerical optimization}
\label{app:protocol_numerical_optimization}

This appendix reports the protocol parameters selected by the fixed-budget
optimization in Section~\ref{sec:two_stage_fixed_budget}. In
Figure~\ref{fig:appendix_protocol_design_all_N2} we show for each value of
$N_2$, the selected Stage~I sample size $N_1^\ast$, the Stage~I squeezing
$r_1^\ast$, the Stage~II squeezing $r_2^\ast$, and the energy fractions $E_1/E$
and $E_2/E$. These quantities are computed from the minimizers of
Eqs.~\eqref{eq:coherent_two_stage_bound_fixed_N2} and
\eqref{eq:dsvs_two_stage_bound_fixed_N2}. %; they are not additional optimization
% objectives.
Figure~\ref{fig:appendix_protocol_design_all_N2} complements the normalized
error bounds in Figure~\ref{fig:two_stage_scaling_comparison} by showing how the
optimized protocol allocates samples, squeezing, and energy between the two
stages.

\begin{figure*}[ht]
  \centering
  \includegraphics[width=.95\textwidth]{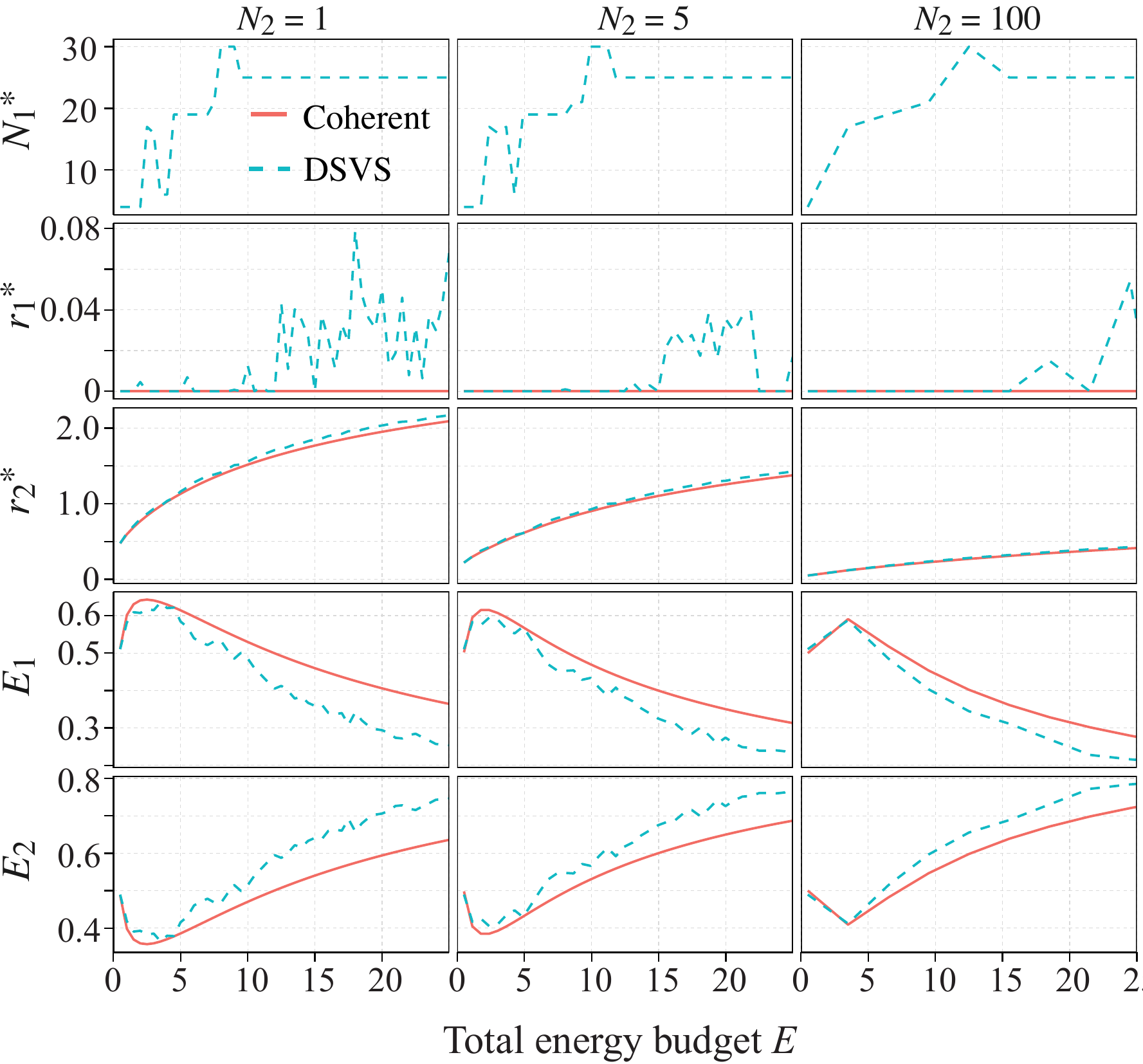}
  \caption{\label{fig:appendix_protocol_design_all_N2} Protocol parameters
    selected by the fixed-budget optimization. The columns correspond to
    $N_2=1$, $N_2=5$, and $N_2=100$. The rows show the selected Stage~I sample
    size $N_1^\ast$, the Stage~I squeezing $r_1^\ast$, the Stage~II squeezing
    $r_2^\ast$, the Stage~I energy fraction $E_1/E$, and the Stage~II energy
    fraction $E_2/E$. For the coherent protocol, $r_1^\ast=0$ by construction.
    The DSVS curves show that the optimized displaced squeezed protocol uses
    only weak Stage~I squeezing over the plotted range, whereas the Stage~II
    squeezing depends strongly on the prescribed value of $N_2$.}
\end{figure*}

%%%%%%%%%%%%%%
% References %
%%%%%%%%%%%%%%

% \nocite{*}
\clearpage
\bibliographystyle{unsrtnat}
\bibliography{References}
\end{document}